\documentclass[journal,twoside,web]{ieeecolor}
\usepackage{generic}

\usepackage{cite}
\usepackage{amsmath,amssymb,amsfonts}
\usepackage{graphicx}
\usepackage{textcomp}

\usepackage{mathtools}
\usepackage{relsize}
\usepackage{tikz}

\usetikzlibrary{
	matrix,
	positioning,
	decorations,
	decorations.pathreplacing
}

\graphicspath{{./Figures/}} 

\usepackage[margin=0.8cm]{caption}
\newcommand{\norm}[1]{\left\lVert#1\right\rVert}
\usepackage{pgfplots}
\pgfplotsset{compat=newest}
\usepackage{tabularx}
\pgfplotsset{plot coordinates/math parser=false}

\usepackage{pgfplots}
\pgfplotsset{compat=newest}
\usepackage{tabularx}
\pgfplotsset{plot coordinates/math parser=false}

\usepackage{algorithm}
\usepackage[noend]{algpseudocode}
\usepackage{soul}
\usepackage{dsfont}
\usepackage{tabto}
\usepackage[normalem]{ulem}

\usepackage{filecontents}
\usetikzlibrary{positioning}
\usetikzlibrary{calc}
\usepackage{siunitx}
\usepackage{cite}

\usepackage{blkarray, bigstrut} %

\newtheorem{lemma}{Lemma}
\newtheorem{theorem}{Theorem}
\newtheorem{assumption}{Assumption}

\newtheorem{remark}{Remark}

\newtheorem{proposition}{Proposition}

\usepackage{lineno,amssymb,subcaption,algpseudocode}

\newcommand{\WS}[1][\epsilon^w]{{\mathbb{W}(#1)}}

\newcommand{\A}{A}
\newcommand{\B}{B}
\newcommand{\C}{C}
\newcommand{\D}{D}

\newcommand{\R}{\mathbb{R}}

\newcommand{\Y}{\mathcal{Y}}

\newcommand{\X}{\mathcal{X}}

\newcommand{\1}{\mathbf{1}}
\newcommand{\0}{\mathbf{0}}
\newcommand{\I}{\mathbf{I}}
\usepackage{accents}

\usepackage{cuted}
\usepackage{amsmath,amssymb,amsfonts}
\usepackage{graphicx}
\usepackage{paralist}
\usepackage{rotating}
\usepackage{verbatim}

\usepackage{tikz}
\usetikzlibrary{calc,patterns,decorations.pathmorphing,decorations.markings,matrix}

\usepackage{nicematrix}

\usepackage{amsmath} 
\usepackage{amssymb}  
\usepackage{mathtools}
\usepackage{graphicx}
\usepackage{dsfont}
\usepackage{bm}
\usepackage{xcolor}
\definecolor{blue_set}{RGB}{204,229.5,255}
\definecolor{pink_set}{RGB}{255,204,229.5}
\definecolor{green_set}{RGB}{229.5,255,204}
\definecolor{grey_set}{RGB}{102,102,102}
\definecolor{cyan_set}{RGB}{69,243,248}
\definecolor{yellow_set}{RGB}{229.5000,229.5000,25.5000}
\definecolor{green_border}{RGB}{0,255,0}
\definecolor{set_blue}{RGB}{25.5000,153.0000,204.0000}
\definecolor{set_green}{RGB}{25.5000,153.0000,25.5000}

\definecolor{safe_set_color}{RGB}{159,2,81}
\definecolor{obtained_color}{RGB}{77,137,99}

\definecolor{red_google}{RGB}{255,62,48}
\definecolor{blue_google}{RGB}{23,107,239}
\definecolor{gray_google}{RGB}{107,107,107}
\definecolor{green_google}{RGB}{77,137,99}

\definecolor{green}{RGB}{0,122,0}
\definecolor{X2_clr}{RGB}{204,204,204}
\definecolor{Y1_clr}{RGB}{40,39,38}
\definecolor{Y2_clr}{RGB}{255,237,102}
\definecolor{W_clr}{RGB}{246,114,128}
\definecolor{X1_clr}{RGB}{0,116,63}
\definecolor{X2_clr}{RGB}{248,177,149}

\usepackage{filecontents}
\usetikzlibrary{positioning}
\usetikzlibrary{calc}
\usepackage{url}
\usepackage{siunitx}

\usepackage{subcaption}

\newcommand\scalemath[2]{\scalebox{#1}{\mbox{\ensuremath{\displaystyle #2}}}}

\definecolor{wheat}{rgb}{0.96,0.87,0.70}
\definecolor{mario}{rgb}{0.8,0.8,1}
\definecolor{SamComm}{rgb}{0.9,0.9,0.1}
\definecolor{ao}{rgb}{0.0, 0.5, 0.0}

\definecolor{olivegreen}{rgb}{ 0.2431,0.4941,0.2431}

\usepackage{fancybox}
\usepackage[most]{tcolorbox}

\newtcbox{\mybox}[1][red]{on line,
	colback=#1, colframe=#1, boxsep=0pt, boxrule=0pt, size=small, arc=1mm}

\usepackage{booktabs}
\usepackage{siunitx}
\usepackage{xcolor}
\usepackage{booktabs,colortbl, array}
\usepackage{pgfplotstable}
\pgfplotsset{compat=1.8}

\definecolor{rulecolor}{RGB}{0,71,171}
\definecolor{tableheadcolor}{gray}{0.92}
\newcommand{\topline}{ %
	\arrayrulecolor{rulecolor}\specialrule{0.1em}{\abovetopsep}{0pt}%
	\arrayrulecolor{tableheadcolor}\specialrule{\belowrulesep}{0pt}{0pt}%
	\arrayrulecolor{rulecolor}}
\newcommand{\midtopline}{ %
	\arrayrulecolor{tableheadcolor}\specialrule{\aboverulesep}{0pt}{0pt}%
	\arrayrulecolor{rulecolor}\specialrule{\lightrulewidth}{0pt}{0pt}%
	\arrayrulecolor{white}\specialrule{\belowrulesep}{0pt}{0pt}%
	\arrayrulecolor{rulecolor}}
\newcommand{\bottomline}{ %
	\arrayrulecolor{white}\specialrule{\aboverulesep}{0pt}{0pt}%
	\arrayrulecolor{rulecolor} %
	\specialrule{\heavyrulewidth}{0pt}{\belowbottomsep}}%

\pgfplotstableset{normal/.style ={%
		header=true,
		string type,
		font=\addfontfeature{Numbers={Monospaced}}\small,
		column type=l,
		every odd row/.style={
			before row=
		},
		every head row/.style={
			before row={\topline\rowcolor{tableheadcolor}},
			after row={\midtopline}
		},
		every last row/.style={
			after row=\bottomline
		},
		col sep=&,
		row sep=\\
	}
}

\newcommand{\mario}[1]{
	\begin{center}
		\fcolorbox{mario}{mario}{\parbox[t]{0.9\linewidth}{\textbf{Mario:} #1}}
\end{center}}

\title{\LARGE \bf
	Computation of safe disturbance sets using implicit RPI sets
}

\author{Sampath Kumar Mulagaleti, Alberto Bemporad, and Mario Zanon 
	\thanks{The authors are with
		the IMT School for Advanced Studies Lucca, Piazza San Francesco 19,
		55100 Lucca, Italy. 
		Email: 
		{\tt\small s.mulagaleti@imtlucca.it}}
}

\begin{document}

	\maketitle
	\thispagestyle{empty}
	\pagestyle{empty}
	
	
	\begin{abstract}
		Given a stable linear time-invariant (LTI) system subject to output constraints, we present a method to compute a set of disturbances such that the reachable set of outputs 
		matches as closely as possible the output constraint set, while being included in it. This problem finds application in several control design problems, such as the development of hierarchical control loops, decentralized control, supervisory control, robustness-verification, etc. We first characterize the set of disturbance sets satisfying the output constraint inclusion using corresponding minimal robust positive invariant (mRPI) sets, following which we formulate an optimization problem that minimizes the distance between the reachable output set and the output constraint set. We tackle the optimization problem using an implicit RPI set approach that provides a priori approximation error guarantees, and adopt a novel disturbance set parameterization that permits the encoding of the set of feasible disturbance sets as a polyhedron. Through extensive numerical examples, we demonstrate that the proposed approach computes disturbance sets with reduced conservativeness improved computational efficiency than state-of-the-art methods.
	\end{abstract}

		\section{Introduction}
		The theory of set invariance provides key tools for the synthesis and analysis of control schemes for dynamical systems~\cite{Blanchini2015}. These tools can be used to construct robust positive invariant (RPI) sets for uncertain dynamical systems, i.e., regions of the state-space in which the system remains for all future times, for all possible uncertainty realizations. RPI sets find application in the design of control schemes for constrained uncertain dynamical systems, such as Robust model predictive control (RMPC) and reference governor schemes \cite{Rawlings2009,Kouvaritakis2015,Garone2017}, fault-tolerant control \cite{Olaru2008,Olaru2010}. Consequently, the development of methods to construct RPI sets is a very active area of research, see, e.g., \cite{Gilbert1991,Kolmanovsky1998,Pluymers2005,Rakovic2005,Rakovic2006,Rakovic2007,Trodden2016}. 
		
		Most existing RPI-based approaches are developed under the assumption that the set of disturbances characterizing the uncertainty set of the system is known a priori.
		In many practical cases, however, while the set of admissible states can be estimated from sensor measurements or given a priori as a set of constraints to be satisfied, the set of disturbances acting on the system might not be related to process noise stemming from unknown dynamics, but it might be possible to use the available information to suitably design it. A relevant example is given by decentralized MPC (DeMPC) application such as~\cite{Riverso2013,Mulagaleti2021}, where the dynamic coupling between subsystems is modeled as an additive disturbance, such that the disturbance set on a given system represents the state-constraint sets of the neighboring subsystems. Another example is presented in \cite{Flores2008}, in which the feedforward reference signal is a disturbance acting on the system, and one needs to design a set of reference signals that permit safe closed-loop operation. Such disturbance set design problems can be appropriately tackled using RPI sets, and we present one such method in this paper using implicit RPI sets.
		
		Given a disturbance set, there exist many techniques to compute RPI sets. Of particular interest in our case is the minimal RPI (mRPI) set, which is the smallest RPI set included in all RPI sets corresponding to a given disturbance set \cite{Blanchini2015}. Since it is generally impossible to obtain an explicit representation of this set for linear systems affected by additive disturbances except under very restrictive assumptions \cite{Seron2019}, many methods typically aim to compute tight RPI outer approximations of this set \cite{Rakovic2005,Rakovic2006,Rakovic2013,Trodden2016,Martinez2015,Tan2019}. There also exist methods to simultaneously synthesize both volume-minimized RPI sets and the corresponding  feedback controllers \cite{Liu2019,Tahir2015,DeLaPena2008,Mulagaleti2022}.
		
		In this paper, we present a method to compute the largest disturbance set for a constrained linear system that
		ensures that the system constraint set is maximally covered by the system reachable set under persistent disturbances; and provides safety guarantees. We exploit the fact that the mRPI set is also the $0$-reachable set \cite{Rakovic2005},
		such that a large disturbance set is one that minimizes the distance between the mRPI set and the system constraint set.
		%
		We previously tackled this problem setup in~\cite{Mulagaleti2020}, in which we considered a polytopically parametrized disturbance set, along with a polytopically parameterized RPI set to tightly approximate the corresponding mRPI set.  We assumed that the normal vectors defining these sets are fixed a priori. 
		We then formulated an optimization problem involving support functions over polytopes to compute a suitable disturbance set. While the approach could compute disturbance sets of an arbitrary shape, the main limitations were that the set of feasible disturbance sets was nonconvex and nonsmooth, thus requiring the development of a specialized optimization algorithm to compute the disturbance sets~\cite{Mulagaleti2022_SmootheningPDIP},
		and the approximation error with respect to the mRPI set could not be specified a priori. In this paper, we present an alternative approach to tackle the problem.
		{\color{black} The key differentiating aspects of the new approach are: (a) adopting the $\mu$-RPI approximation of the mRPI set proposed in~\cite{Rakovic2006} and permitting the approximation error $\mu$ to be specified a priori, and b) parameterizing the disturbance set as a convex hull of boxes, overcoming the limitation having to define the normal vectors of the set a priori. These two choices together allow the expression of the set of feasible disturbance sets using linear inequalities that are easy to resolve.}
		%
		%
		Finally, the optimization problem that we formulate is smooth and hence can be tackled using off-the-shelf nonlinear programing solvers. Exploiting the structure of the problem, we present an initialization technique based on Linear Programing, and also present several practical heuristics that we demonstrate to be effective through numerical examples.
		In our methods, we do not explicitly compute a representation RPI set, thus referring to the current approach as \textit{implicit}. 
		
		{\color{black} A few approaches have been proposed previously for computing disturbance sets for constrained linear systems. In the method proposed in~\cite{DeSantis1994}, the largest disturbance set that renders a known positive invariant set robust is computed. A similar approach is used in~\cite{Kalabic2012} for synthesizing reduced-order reference governors, where a positive invariant set is first computed and then made RPI with the maximal disturbance set. In contrast, our method co-synthesizes an RPI set (implicitly) and the corresponding disturbance set (explicitly) using the system constraint set.
			An alternative approach, proposed in~\cite{Darup2016,Yang2020}, characterizes the largest scaling factor for a given disturbance set that ensures the existence of a constraint-satisfying RPI set. In contrast, our approach synthesizes arbitrarily-shaped polytopic disturbance sets, resulting in reduced conservativeness.}
		%
		%
		
		The paper is organized as follows. In Section~\ref{sec:prob_def}, we define the problem tackled, and highlight the issues associated with solving it exactly. We then develop a suitable approximation of the problem in Section~\ref{sec:implicit_RPI_set} using implicit RPI sets, and present an optimization algorithm to solve the resulting problem in Section~\ref{sec:optimization_algorithm}. In Section~\ref{sec:approx_soln_methods}, we present heuristic-based LP approximations of the problem formulation.
		{\color{black} In Section~\ref{sec:numerical_examples}, we first present an illustrative example, and then compare our method with the approaches of~\cite{Mulagaleti2020} and~\cite{DeSantis1994} over randomly generated examples. Finally, we apply the method to synthesize constraint sets for reduced-order control schemes.}
		
		\section{Preliminaries} 
		Consider sets $\mathcal{X}, \mathcal{Y} \subset \R^{n}$, and vectors $a \in \R^{n_a}$ and $b \in \R^{n_b}$.
		Given a matrix $L\in \R^{m \times n}$, we denote by $L\mathcal{X}$ the image $\{y\in\R^{m}: y=Lx, x \in \mathcal{X}\}$ of $\mathcal{X}$ under the linear transformation induced by $L$, and $|L|$
		denotes the matrix obtained by element-wise application of the absolute value operator on matrix $L$.
		We denote the $i$-th row of matrix $L$ by $L_i$, an element in row $i$ and column $j$ by $L_{ij}$.
		Given a square matrix $L \in \R^{n \times n}$, $\rho(L)$ denotes its spectral radius. 
		The set $\mathcal{B}_p^n:=\{x:\norm{x}_{p}\leq 1\}$ is the unit $p$-norm ball. 
		Given two matrices $L$, $M\in \R^{n \times m}$, $L \leq M$ denotes element-wise inequality. 
		The symbols $\1$ and $\0$ denote all-ones and all-zeros vectors respectively, and $\I$ denotes the identity matrix with dimensions specified if the context is ambiguous.
		The set of natural numbers between two integers $m$ and $n$, $m\leq n$, is denoted by $\mathbb{I}_m^n:=\{m,\ldots,n\}$.
		The Minkowski set addition is defined as $\mathcal{X} \oplus \mathcal{Y}:=\{x+y:x\in\mathcal{X},y\in\mathcal{Y}\}$, {\color{black} and set subtraction as $\mathcal{X} \ominus \mathcal{Y}:=\{x:\{x\} \oplus \mathcal{Y} \subseteq \mathcal{X}\}$. If $\mathcal{Y} \nsubseteq \mathcal{X}$, then $\mathcal{X} \ominus \mathcal{Y}$ is the empty set.}
		Given a polytope $\mathcal{X}$, we denote its vertices by $\mathrm{vert}(\mathcal{X})$, and given compact sets $\{\mathcal{X}_i, i \in \mathbb{I}_1^N\}$, we denote the convex-hull of these sets by $\mathrm{ConvHull}(\mathcal{X}_i, i \in \mathbb{I}_1^N)$.
	Given compact set $\mathcal{S} \subset \R^{n}$, matrix $\mathbf{T} \in \R^{l \times n}$ and vector  $\mathbf{p} \in \R^{l\times 1}$, the support function is defined as
	\begin{align}
		\label{eq:SF_prop_1}
		h_{\mathbf{T}\mathcal{S}}(\mathbf{p}):=\max_{z \in \mathbf{T}\mathcal{S}} \ \mathbf{p}^{\top} z=\max_{w \in \mathcal{S}} \ \mathbf{p}^{\top} \mathbf{T} w.
	\end{align}
	For compact sets $\mathcal{S}_1,\mathcal{S}_2 \subset \R^n$, support functions satisfy \cite{Kolmanovsky1998}
	\begin{align}
		\label{eq:SF_prop_3}
		\hspace{-6pt}
		h_{\mathcal{S}_1 \oplus \mathcal{S}_2}(\mathbf{p})=h_{\mathcal{S}_1}(\mathbf{p})+h_{\mathcal{S}_2}(\mathbf{p}),
	\end{align}
	and the inclusion $\mathcal{S}_1 \subseteq \mathcal{S}_2$ holds if and only if
	\begin{align}
		\label{eq:set_inclusion_generic_compact}
		h_{\mathcal{S}_1}(\mathbf{p}) \leq h_{\mathcal{S}_2}(\mathbf{p}), && \forall \ \mathbf{p} \in \R^n.
	\end{align}
	Given a matrix $\mathbf{M} \in \R^{q \times n}$, we denote by  $h_{\mathbf{T}\mathcal{S}}(\mathbf{M})$ the $q$-dimensional column vector
	with elements $h_{\mathbf{T}\mathcal{S}}(\mathbf{M}_i^{\top})$, i.e.,
	$$h_{\mathbf{T}\mathcal{S}}(\mathbf{M}):=[h_{\mathbf{T}\mathcal{S}}(\mathbf{M}_1^{\top})  \ \cdots \ h_{\mathbf{T}\mathcal{S}}(\mathbf{M}_q^{\top})]^{\top}.$$
	For polytope $\mathcal{R}:=\{w:\mathbf{M} w \leq \mathbf{t}\}$ and compact set $\mathcal{S}$ in $\R^n$, 
	\begin{align}
		\label{eq:SF_prop_2}
		\mathcal{S} \subseteq \mathcal{R} \iff h_{\mathcal{S}}(\mathbf{M}) \leq h_{\mathcal{R}}(\mathbf{M}) \leq \mathbf{t},
	\end{align}
	Given $\mathcal{S}=\bar{z}\oplus\{z:-\epsilon^z \leq z \leq \epsilon^z\}$ shaped as a box with $\bar{z},\epsilon^z \in \R^{n}$, the following property holds for the support function \cite[Chapter 6]{Blanchini2015}:
	\begin{align}
		\label{eq:box_SF}
		h_{\mathbf{T}\mathcal{S}}(\mathbf{p})= \mathbf{p}^{\top} \mathbf{T}\bar{z} + |\mathbf{p}^{\top} \mathbf{T}|\epsilon^z.
	\end{align}
	We recall some basic properties of set operations.
	\begin{proposition}\cite{Schneider2013}
		\label{prop:basic_properties}
		Let $\mathcal{X},\mathcal{Y},\mathcal{Z} \subset \R^n$ be any compact and convex sets containing the origin, $M \in \R^{m \times n}$ be any matrix of adequate dimension, and $\alpha\geq \beta>0$ be any scalars. 
		\\ ($a$) $\mathcal{X} \subseteq \mathcal{Y} \Rightarrow M\mathcal{X} \subseteq M\mathcal{Y}$; \ \ ($b$) $\mathcal{X} \subseteq \mathcal{Y} \Leftrightarrow \mathcal{X} \oplus \mathcal{Z} \subseteq \mathcal{Y} \oplus \mathcal{Z}$; \\ ($c$) $\alpha \mathcal{X} \oplus \beta \mathcal{X} = (\alpha+\beta) \mathcal{X}$; \ \ ($d$) $\mathcal{X} \subseteq \mathcal{X} \oplus \mathcal{Y}$; \\
		$(e)$ If $\mathcal{X}=\mathrm{ConvHull}(x_{[i]},i \in \mathbb{I}_1^{v})$, then $\mathcal{X} \subseteq \mathcal{Y} \oplus \mathcal{Z}$ holds if and only if for each $i \in \mathbb{I}_1^v$, there exist $y_{[i]} \in \mathcal{Y}$ and $z_{[i]} \in \mathcal{Z}$ such that $x_{[i]}=y_{[i]}+z_{[i]}$.
	\end{proposition}
	\section{Problem Definition}
	\label{sec:prob_def}
	Consider the linear time-invariant discrete-time system
	\begin{subequations}
		\label{eq:system}
		\begin{align} 
			{\scalemath{1}{x(t+1)}}&={\scalemath{0.9}{\A x(t)+\B w(t),}} \label{eq:system:state}\\
			{\scalemath{1}{y(t)}}&={\scalemath{0.9}{\C x(t)+\D w(t),}} \label{eq:system:output}
		\end{align}
	\end{subequations}
	with state $x\in\R^{n_x}$, output $y\in \R^{n_y}$ and  additive disturbance $w\in\R^{n_w}$.
	We assume that a set of output constraints is given:
	\begin{align}
		\label{eq:Y_orig_defn}
		\Y:=\{y:Gy \leq g\} && \text{with} \ g \in \R^{m_Y}.
	\end{align}
	We define the reachable set of states from the origin, i.e., from $x(0)=\0$, under the action of disturbances $w(t) \in \mathcal{W}$ for all $t \geq 0$ in $t$-time steps as
	\begin{align*}
		\mathcal{X}(t,\mathcal{W}):=\left\{x: x=\sum_{k=0}^{t-1} A^{t-k-1} B w(k), \ \forall \ w(k) \in \mathcal{W}\right\},
	\end{align*}
	and the corresponding set of $t$-step reachable outputs as
	\begin{align*}
		\mathcal{Y}(t,\mathcal{W}):=C\mathcal{X}(t,\mathcal{W}) \oplus D \mathcal{W}.
	\end{align*}
	Observing that the reachable set of states satisfies the inclusion
	\begin{align*}
		\mathcal{X}(t,\mathcal{W}) \subseteq \mathcal{X}(t+1,\mathcal{W}), && \forall \ t\geq 0,
	\end{align*}
	if the disturbance set $\mathcal{W}$ is compact, convex, and contains the origin, we define the limit of reachable set of states as
	\begin{align}
		\label{eq:mRPI_limit_definition}
		\mathcal{X}_{\mathrm{m}}(\mathcal{W}):=\lim_{t \to \infty} \mathcal{X}(t,\mathcal{W}),
	\end{align}
	and the corresponding limit set of reachable outputs as
	\begin{align}
		\label{eq:Ym_mRPI_defn}
		\mathcal{Y}_{\mathrm{m}}(\mathcal{W}):=\C \mathcal{X}_{\mathrm{m}}(\mathcal{W})\oplus \D \mathcal{W}.	
	\end{align}
	Then, our goal is to compute a disturbance set $\mathcal{W}$ that satisfies
	\begin{align}
		\label{eq:ideal_reachability}
		{\color{black}	\mathcal{Y}_{\mathrm{m}}(\mathcal{W})=\mathcal{Y},}
	\end{align}
	i.e., the reachable set of outputs is equal to the assigned set of outputs $\mathcal{Y}$.
	Satisfying the equality in~\eqref{eq:ideal_reachability} exactly may not be feasible, as $\mathcal{Y}$ is user-specified and can have an arbitrary shape. See~\cite{Mulagaleti2020} for further details. Hence, we instead focus on computing a disturbance set $\mathcal{W}$ that satisfies the inclusion.
	\begin{align}
		\label{eq:safety_inclusion}
		\mathcal{Y}_{\mathrm{m}}(\mathcal{W}) \subseteq \mathcal{Y},
	\end{align}
	and minimizes the distance between the reachable set of outputs $	\mathcal{Y}_{\mathrm{m}}(\mathcal{W})$ and the assigned output set $\mathcal{Y}$. To this end, tackle the optimization problem
	\begin{subequations}
		\label{eq:orig_problem_to_solve}
		\begin{align}
			\min_{\mathcal{W}}\hspace{5pt} & \mathrm{d}_{\mathcal{Y}}(\mathcal{Y}_{\mathrm{m}}(\mathcal{W})) 	\label{eq:orig_problem_to_solve:obj}\\ \text{ s.t. } & \mathcal{Y}_{\mathrm{m}}(\mathcal{W}) \subseteq \mathcal{Y}, 	\label{eq:orig_problem_to_solve:con_1}
			\\ & \0 \in \mathcal{W}, \label{eq:orig_problem_to_solve:con_2}
		\end{align}
	\end{subequations}
	where $\mathrm{d}_{\mathcal{Y}}(\mathcal{Y}_{\mathrm{m}}(\mathcal{W}))$ measures the disturbance between the sets $\mathcal{Y}$ and $\mathcal{Y}_{\mathrm{m}}(\mathcal{W})$, and Constraint~\eqref{eq:orig_problem_to_solve:con_1} enforces the desired inclusion in~\eqref{eq:safety_inclusion}.The \textit{outer-inclusion} problem with constraint $\mathcal{Y}_{\mathrm{m}}(\mathcal{W}) \supseteq \mathcal{Y}$ can be solved with minor modifications to the techniques presented in the sequel. However, for simplicity of exposition, we discuss the \textit{inner-inclusion} setting in Problem~\eqref{eq:orig_problem_to_solve}.
	{\color{black} We present effective approaches for tackling Problem~\eqref{eq:orig_problem_to_solve} by focusing on developing suitable parameterizations and approximations of the disturbance set $\mathcal{W}$ and the associated output reachable set $\mathcal{Y}_{\mathrm{m}}(\mathcal{W})$. By doing so, we can reformulate Problem~\eqref{eq:orig_problem_to_solve} into a form that can be tackled using standard numerical optimization solvers.
		\begin{remark}
			Problem~\eqref{eq:orig_problem_to_solve} can handle scenarios with independent process and measurement disturbances by selecting matrices $C$ and $D$ appropriately, and defining $\mathcal{W}$ as the Cartesian product of these disturbance sets. $\hfill\square$
		\end{remark}
	}	
	Regarding the distance function $\mathrm{d}_{\mathcal{Y}}(\mathcal{Y}_{\mathrm{m}}(\mathcal{W}))$, a classical choice is to use the Hausdorff distance between the sets $\mathcal{Y}_{\mathrm{m}}(\mathcal{W})$ and $\mathcal{Y}$~\cite{Mulagaleti2020}. In this paper, we consider a slightly more general formulation with	
	\begin{align}
		\label{eq:objective_exact}
		\mathrm{d}_{\mathcal{Y}}(\mathcal{Y}_{\mathrm{m}}(\mathcal{W})) := \min_{\epsilon} \{\norm{\epsilon}_1 : \mathcal{Y} \subseteq \mathcal{Y}_{\mathrm{m}}(\mathcal{W}) \oplus \mathbb{B}(\epsilon) \}
	\end{align}
	defined using the set $\mathbb{B}(\epsilon):=\{y:Hy \leq \epsilon\}$ {\color{black}with $H \in \R^{n_B \times n_y}$ and $\epsilon \in \R^{n_B}$}, in which normal vectors $\{H_i^{\top}\}$ are specified a priori by the user.
	Since $\mathrm{d}_{\mathcal{Y}}(\cdot)$ is monotonic, i.e, for all compact sets $\mathbf{S}_1,\mathbf{S}_2 \subseteq \mathcal{Y}$,
	\begin{align}
		\label{eq:monotonicty_property}
		\mathbf{S}_1 \subseteq \mathbf{S}_2 \subseteq \mathcal{Y}  \implies \mathrm{d}_{\mathcal{Y}}(\mathbf{S}_1) \geq \mathrm{d}_{\mathcal{Y}}(\mathbf{S}_2),
	\end{align}
	Problem~\eqref{eq:orig_problem_to_solve} computes a disturbance set $\mathcal{W}$ that maximizes the coverage of $\mathcal{Y}$ by the reachable outputs while enforcing inclusion~\eqref{eq:safety_inclusion}.
	Moreover, it prioritizes coverage of $\mathcal{Y}$ by the output reachable set $\mathcal{Y}_{\mathrm{m}}(\mathcal{W})$ in directions indicated by the normal vectors $H_i^{\top}$ to the set $\mathbb{B}(\epsilon)$. {\color{black} The key difference between using the standard Hausdorff distance and this method is that these normal vectors are not restricted to a symmetric set $\mathbb{B}(\epsilon)$.}

	{\color{black} \begin{remark}
			A typical choice for matrix $H$ would be to define it using the normal vectors of some $p$-norm ball.
			In some applications however, only constraint enforcement is desired in some output directions, while maximal reachability in the others. This can be achieved by selecting the normal vectors of $\mathbb{B}(\epsilon)$ appropriately, e.g., $\mathbb{B}(\epsilon)$ can be parameterized as a reduced-dimensional polytope in $\R^{n_y}$. $\hfill\square$
		\end{remark} 
	\begin{remark}
		In Problem~\eqref{eq:orig_problem_to_solve}, alternative objective functions such as the maximization of the volume of $\mathcal{W}$ can also be explored, the encoding of which relies heavily on the disturbance set parametrization. The investigation and development of such approaches, and comparison with the methodology proposed in this paper is a subject of future study.
	$\hfill\square$
		\end{remark}
	}

	Before tackling Problem~\eqref{eq:orig_problem_to_solve}, we observe that the reachable state set $\mathcal{X}_{\mathrm{m}}(\mathcal{W})$ in~\eqref{eq:mRPI_limit_definition} is given by the infinite Minkowski sum
	\begin{align}
		\label{eq:mRPI}
		\X_{\mathrm{m}}(\mathcal{W})= \bigoplus_{t=0}^{\infty} \A^t\B\mathcal{W}.
	\end{align}
	This implies that, except under very restrictive assumptions on $(A,B,\mathcal{W})$~\cite{Seron2019}, computing an exact finite-dimensional representation of the set $\mathcal{X}_{\mathrm{m}}(\mathcal{W})$, and hence $\mathcal{Y}_{\mathrm{m}}(\mathcal{W})$, is in general impossible. Thus, Problem~\eqref{eq:orig_problem_to_solve} is in general impossible to solve exactly. This can, however, be ameliorated by adopting the notion of RPI sets, as we now explain. To this end, we make the following standing assumption of 
	System~\eqref{eq:system}.
	\begin{assumption}
		\label{ass:stable}
		System \eqref{eq:system} is strictly stable, i.e., $\rho(A)<1$. $\hfill\square$
	\end{assumption} 
	
	Under Assumption~\ref{ass:stable}, we know from~\cite{Kolmanovsky1998} that if the disturbance set $\mathcal{W}$ is compact, convex and contains the origin, then $\mathcal{X}_{\mathrm{m}}(\mathcal{W})$ and $\mathcal{Y}_{\mathrm{m}}(\mathcal{W})$ exist, are compact, convex, and contain the origin. Moreover, as shown in~\cite{Kolmanovsky1998}, $\mathcal{X}_{\mathrm{m}}(\mathcal{W})$ is the smallest (in an inclusion sense) RPI set of System~\eqref{eq:system:state}, i.e., if a set $\mathcal{X}_{\mathrm{RPI}}(\mathcal{W})$ satisfies the RPI inclusion
	\begin{align}
		\label{eq:RPI_inclusion}
		A \mathcal{X}_{\mathrm{RPI}}(\mathcal{W}) \oplus B \mathcal{W} \subseteq \mathcal{X}_{\mathrm{RPI}}(\mathcal{W}),
	\end{align}
	then the inclusion $\mathcal{X}_{\mathrm{m}}(\mathcal{W}) \subseteq  \mathcal{X}_{\mathrm{RPI}}(\mathcal{W})$  follows.
	Hence, 	$\mathcal{X}_{\mathrm{m}}(\mathcal{W})$ is also referred to as the minimal RPI (mRPI) set. Then, defining a set of outputs corresponding to some RPI set $\mathcal{X}_{\mathrm{RPI}}(\mathcal{W})$ as $	\mathcal{Y}_{\mathrm{RPI}}(\mathcal{W}):=C\mathcal{X}_{\mathrm{RPI}}(\mathcal{W}) \oplus D\mathcal{W} \supseteq \mathcal{Y}_{\mathrm{m}}(\mathcal{W}),$
	the desired output constraint inclusion in~\eqref{eq:safety_inclusion} formulating Constraint~\eqref{eq:orig_problem_to_solve:con_1} can be enforced through
	\begin{align}
		\label{eq:RPI_output_basic}
		\mathcal{Y}_{\mathrm{RPI}}(\mathcal{W}) \subseteq \mathcal{Y}.
	\end{align}
	Thus, a typical approach to tackle Problem~\eqref{eq:orig_problem_to_solve} involves replacing the output reachable set $\mathcal{Y}_{\mathrm{m}}(\mathcal{W})$ with some outer-approximating set $\mathcal{Y}_{\mathrm{RPI}}(\mathcal{W})$ defined using some suitable finite-dimensional RPI set $\mathcal{X}_{\mathrm{RPI}}(\mathcal{W})$. 
	
	Given a disturbance set $\mathcal{W}$, there exist many techniques, e.g., \cite{Rakovic2005,Rakovic2006,Rakovic2013,Trodden2016}, to compute tight outer RPI approximations of the mRPI set. Thus, our contribution is directed towards embedding such techniques within an optimization problem framework for concurrent disturbance set and RPI set computation. 
	In \cite{Mulagaleti2020}, we adopted a polytopic RPI set parameterization with fixed normal vectors from \cite{Rakovic2013,Trodden2016} to approximate the mRPI set in Problem~\eqref{eq:orig_problem_to_solve}, and also parameterized the disturbance set as a polytope with fixed normal vectors. While this approach permitted the computation of disturbance sets of an arbitrary shape, it did not allow for a priori specification of approximation error bounds and resulted in a nonconvex set of disturbance sets verifying inclusion~\eqref{eq:RPI_output_basic}. In this paper, we propose a new methodology that adopts the RPI set parameterization from \cite{Rakovic2005,Rakovic2006} to approximate the mRPI set, and a novel disturbance set parameterization based on convex hull of boxes. This approach allows us to represent sets of arbitrary complexity without the need for prior specification of normal vectors, and the RPI set parameterization lets us fix the approximation error a priori. The resulting set of disturbance sets verifying inclusion~\eqref{eq:RPI_output_basic} is represented using linear inequalities. However, we make the slightly restrictive assumption that the output constraint set $\mathcal{Y}$ is more restrictive than necessary in~\cite{Mulagaleti2020}.
	\begin{assumption}
		\label{ass:0_in_Y_set}
		The output constraint set $\mathcal{Y}$ contains the origin in its nonempty interior, such that $g>\0$.  $\hfill\square$
	\end{assumption}
	We refer to our approach as using \textit{implicit} RPI sets, since
	unlike in~\cite{Mulagaleti2020}, we do not explicitly compute a representation of the RPI set, but rather use an implicit representation parameterized by the disturbance set.
	\begin{remark}
		The developments can be adapted for output sets $\mathcal{Y}$ of the form $\bar{y}_0 \oplus \bar{\mathcal{Y}}$, with $\bar{\mathcal{Y}}$ being a polytope containing the origin in its nonempty interior and $\bar{y}_0$ such that $\mathcal{Y}$ does not contain the origin, violating Assumption~\ref{ass:0_in_Y_set}. The disturbance set $\mathcal{W}$ can then be represented as $\bar{w}_0 \oplus \bar{\mathcal{W}}$, with $\bar{w}_0$ and some $\bar{x}_0$ satisfying the steady-state condition $\bar{x}_0 = A\bar{x}_0 + B \bar{w}_0$ and $\bar{y}_0=C\bar{x}_0 + D\bar{w}_0$.
		%
		The values of $(\bar{x}_0,\bar{w}_0)$ can be pre-computed by solving these conditions, and $\bar{\mathcal{W}}$ can be computed by Problem~\eqref{eq:orig_problem_to_solve} by replacing $\mathcal{Y}$ with $\bar{\mathcal{Y}}$. Alternatively, $(\bar{x}_0,\bar{w}_0)$ can be computed along with $\bar{\mathcal{W}}$  by introducing the steady-state conditions as constraints into the problem.  $\hfill\square$
	\end{remark}
\section{Implicit RPI set parametrization}
\label{sec:implicit_RPI_set}
%
%
In this section, we approximate Problem~\eqref{eq:orig_problem_to_solve} using arbitrarily tight RPI approximations of the mRPI set $\mathcal{X}_{\mathrm{m}}(\mathcal{W})$. To that end, we will rely on the following definition.

\noindent
{\textbf{\textit{Definition}}}:
Given a disturbance set $\mathcal{W}$, for a given $\mu>0$, an RPI set $\mathcal{X}_{\mathrm{\mu}}(\mathcal{W})$ is a $\mu$-RPI set for System \eqref{eq:system:state} if 
it satisfies the inclusions 
\begin{align}
\hspace{20pt}
\label{eq:muRPI_inclusions}
\mathcal{X}_{\mathrm{m}}(\mathcal{W}) \ \subseteq \  
\mathcal{X}_{\mu}(\mathcal{W}) \ \subseteq \	\mathcal{X}_{\mathrm{m}}(\mathcal{W}) \oplus \mu \mathcal{B}_{\infty}^{n_x}. \qquad  \hfill\square
\end{align}%
To allow for a priori specification of the desired approximation error, we will use property~\eqref{eq:muRPI_inclusions} of $\mu$-RPI sets to approximate the mRPI set $\mathcal{X}_{\mathrm{m}}(\mathcal{W})$ in Problem~\eqref{eq:orig_problem_to_solve} with a specified error bound.
{\color{black} This error bound dictates the level of conservativeness introduced by approximating Problem~\eqref{eq:orig_problem_to_solve} using $\mathcal{X}_{\mu}(\mathcal{W})$. In order to observe this, we define the sets
\begin{subequations}
	\label{eq:S1S2S3_defn}
	\begin{align}
		&\mathcal{S}_1:=\{\0 \in \mathcal{W} :  C\mathcal{X}_{\mathrm{m}}(\mathcal{W})\oplus D\mathcal{W}\subseteq \mathcal{Y} \}, \label{eq:s1_defn}\\
		&\mathcal{S}_2:=\{\0 \in \mathcal{W} : C\mathcal{X}_{\mu}(\mathcal{W}) \oplus D\mathcal{W}\subseteq \mathcal{Y} \}, \label{eq:s2_defn}\\
		&\mathcal{S}_3:=\{\0 \in \mathcal{W} :  C\mathcal{X}_{\mathrm{m}}(\mathcal{W})\oplus D\mathcal{W} \oplus {\mu} C \mathcal{B}_{\infty}^{n_x} \subseteq \mathcal{Y} \}. \label{eq:s3_defn}
	\end{align}
\end{subequations}
	For any user-specified $\mu>0$, these sets satisfy $\mathcal{S}_3 \subseteq \mathcal{S}_2 \subseteq \mathcal{S}_1$ as per the inclusions in~\eqref{eq:muRPI_inclusions}. Defining some distance metric $\mathbf{D}(\cdot,\cdot)$ over these sets, the inequality $\mathbf{D}(\mathcal{S}_1,\mathcal{S}_2) \leq \mathbf{D}(\mathcal{S}_1,\mathcal{S}_3)$ follows. Furthermore, $\mathcal{S}_1$ is the feasible set of Problem~\eqref{eq:orig_problem_to_solve}, and $\mathcal{S}_2$ is the feasible set obtained upon approximating the mRPI set $\mathcal{X}_{\mathrm{m}}(\mathcal{W})$ in Problem~\eqref{eq:orig_problem_to_solve} with the $\mu$-RPI set  $\mathcal{X}_{\mu}(\mathcal{W})$. The conservativeness in this approximation is characterized by $\mathbf{D}(\mathcal{S}_1,\mathcal{S}_2)$, that is upper-bounded by $\mathbf{D}(\mathcal{S}_1,\mathcal{S}_3)$ and dictated by $\mu$. 
	Thus, a small $\mu>0$ reduces $\mathbf{D}(\mathcal{S}_1,\mathcal{S}_3)$, such that $\mathcal{S}_2$ becomes a tighter approximation of $\mathcal{S}_1$.
	
	We will now recall a result from \cite{Rakovic2006} to compute a $\mu$-RPI set that, however, requires knowing the disturbance set $\mathcal{W}$.
	As we do not have the set $\mathcal{W}$ a priori, we will need to develop additional theoretical results to exploit it.
}
\begin{lemma}\cite[Section III-B]{Rakovic2006}
\label{lemma:original_Rakovic} Suppose that Assumption~\ref{ass:stable} holds. For some index $s>0$, scalars $\alpha \in [0,1)$, $\lambda \in [0,1]$, and disturbance set $\mathcal{W}$ such that $\0 \in \mathcal{W}$, if
\begin{align}
\label{eq:RPI_condition_check}
\scalemath{1}{
A^s(B \mathcal{W} \oplus \lambda \mathcal{B}^{n_x}_{\infty})\subseteq \alpha(B\mathcal{W} \oplus \lambda \mathcal{B}_{\infty}^{n_x})}
\end{align}
holds, then the parametrized set
\begin{align}
\label{eq:mu_RPI_set}
\mathcal{R}(s,\alpha,\lambda,\mathcal{W}):=(1-\alpha)^{-1} \bigoplus_{t=0}^{s-1} A^t(B\mathcal{W} \oplus \lambda \mathcal{B}_{\infty}^{n_x})
\end{align}
is RPI for System \eqref{eq:system:state} with persistent disturbances $w \in \mathcal{W}$. Moreover, if for some scalar $\mu>0$, the inclusion
\begin{align}
\label{eq:distance_condition}
(1-\alpha)^{-1}\bigoplus_{t=0}^{s-1}A^t(\alpha B\mathcal{W} \oplus \lambda \mathcal{B}_{\infty}^{n_x})\subseteq \mu \mathcal{B}_{\infty}^{n_x}
\end{align}
holds, then $\mathcal{R}(s,\alpha,\lambda,\mathcal{W})$ is a $\mu$-RPI set.
%
$\hfill\square$
\end{lemma}

We briefly recall the rationale behind Lemma~\ref{lemma:original_Rakovic}. 
For any compact and convex disturbance set $\mathcal{W}$ containing the origin and for any index $s>0$, the inclusion
\begin{align*}
\mathcal{X}(s,\mathcal{W}):=\bigoplus_{t=0}^{s-1} A^t B \mathcal{W} \subseteq \bigoplus_{t=0}^{\infty} A^tB \mathcal{W} \overset{\eqref{eq:mRPI}}{=} \mathcal{X}_{\mathrm{m}}(\mathcal{W})
\end{align*}
holds. Then, we know from~\cite[Theorem 1]{Rakovic2005} that if the inclusion  $A^sB\mathcal{W}\subseteq \alpha B\mathcal{W}$ holds for some index $s>0$ and a scalar $\alpha \in [0,1)$, 
the scaled set $(1-\alpha)^{-1} \mathcal{X}(s,\mathcal{W})$ is RPI for System~\eqref{eq:system:state} with disturbance set $\mathcal{W}$.
However, as noted in~\cite{Rakovic2006}, there might not exist some parameters $(s,\alpha)$ satisfying $A^sB\mathcal{W}\subseteq \alpha B\mathcal{W}$ unless the origin is included in the interior of the set $B\mathcal{W}$. This can occur, for example, if the rank of matrix $B$ is smaller than $n_x$, thus posing a structural restriction on the applicability of~\cite[Theorem 1]{Rakovic2005}. 
To overcome this limitation, a modification proposed in \cite{Rakovic2006} involves perturbing the disturbance set with some $\lambda \in [0,1]$ as
$\tilde{\mathcal{W}}(\lambda):=B\mathcal{W}\oplus \lambda\mathcal{B}_{\infty}^{n_x}.$
Then, since the interior of $\tilde{\mathcal{W}}(\lambda)$ is nonempty for every $\lambda \in (0,1]$, there always exists some $(s,\alpha)$ satisfying the inclusion in~\eqref{eq:RPI_condition_check}, i.e., $A^s \tilde{\mathcal{W}}(\lambda) \subseteq \alpha \tilde{\mathcal{W}}(\lambda).$
From~\cite[Theorem 1]{Rakovic2005}, it then follows that the set $\mathcal{R}(s,\alpha,\lambda,\mathcal{W})$ defined in~\eqref{eq:mu_RPI_set} is RPI for the modified system $x(t+1)=Ax(t)+\tilde{w}(t)$
with disturbance set $\tilde{\mathcal{W}}(\lambda)$. Finally, the result in Lemma~\ref{lemma:original_Rakovic} follows from~\cite[Lemma 2]{Rakovic2006}, which states that every RPI set of the modified system with disturbance set $\tilde{\mathcal{W}}(\lambda)$ is an RPI set for System~\eqref{eq:system:state} with disturbance set $\mathcal{W}$. For the reasoning behind inclusion~\eqref{eq:distance_condition} that leads to $\mathcal{R}(s,\alpha,\lambda,\mathcal{W})$  being a $\mu$-RPI set, the reader is referred to~\cite[Theorem 3]{Rakovic2006}.
%

Since our aim is to construct a $\mu$-RPI set approximating the mRPI set $\mathcal{X}_{\mathrm{m}}(\mathcal{W})$, Lemma~\ref{lemma:original_Rakovic} provides us with the important result that the parameterized set $\mathcal{R}(s,\alpha,\lambda,\mathcal{W})$ is $\mu$-RPI. We then use the parameterized set $\mathcal{R}(s,\alpha,\lambda,\mathcal{W})$ to approximate the output reachable set $\mathcal{Y}_{\mathrm{m}}(\mathcal{W})$ as
\begin{align}
\label{eq:O_outside_Y}
\mathcal{O}(s,\alpha,\lambda,\mathcal{W}):=C\mathcal{R}(s,\alpha,\lambda,\mathcal{W}) \oplus D\mathcal{W} \ \supseteq \ \mathcal{Y}_{\mathrm{m}}(\mathcal{W}).
\end{align}
Since $\mathcal{O}(s,\alpha,\lambda,\mathcal{W})$ is an outer approximation of the output reachable set $\mathcal{Y}_{\mathrm{m}}(\mathcal{W})$, the inclusion 
\begin{align}
\label{eq:output_inclusion_O}
\mathcal{O}(s,\alpha,\lambda,\mathcal{W}) \subseteq \mathcal{Y}
\end{align}
entails the desired output inclusion $\mathcal{Y}_{\mathrm{m}}(\mathcal{W})\subseteq \mathcal{Y}$ in~\eqref{eq:orig_problem_to_solve:con_1}.
Hence, we propose to replace the constraint~\eqref{eq:orig_problem_to_solve:con_1} with~\eqref{eq:output_inclusion_O}. 
Regarding the objective $\mathrm{d}_{\mathcal{Y}}(\mathcal{Y}_{\mathrm{m}}(\mathcal{W}))$
{\color{black}however, replacing the output reachable set $\mathcal{Y}_{\mathrm{m}}(\mathcal{W})$ by the approximation ${O}(s,\alpha,\lambda,\mathcal{W})$ would imply minimizing a lower bound to Problem~\eqref{eq:orig_problem_to_solve}. This is because the inequality $$\mathrm{d}_{\mathcal{Y}}(\mathcal{O}(s,\alpha,\lambda,\mathcal{W})) \leq \mathrm{d}_{\mathcal{Y}}(\mathcal{Y}_{\mathrm{m}}(\mathcal{W}))$$ follows from monotonicity of $\mathrm{d}_{\mathcal{Y}}(\cdot)$ in~\eqref{eq:monotonicty_property} since the inclusion  $\mathcal{Y}_{\mathrm{m}}(\mathcal{W}) \subseteq \mathcal{O}(s,\alpha,\lambda,\mathcal{W})$ holds.
Hence, we propose to instead consider the 
the $l$-step output reachable set defined using some user-specified $l>0$,
\begin{align}
\label{eq:l_step_reachable}
\mathcal{S}(l,\mathcal{W}):=\bigoplus_{t=0}^{l-1} CA^tB\mathcal{W} \oplus D\mathcal{W},
\end{align}
and approximate the objective as $\mathrm{d}_{\mathcal{Y}}(\mathcal{S}(l,\mathcal{W}))$. 
}Since the inclusion $\mathcal{S}(l,\mathcal{W}) \subseteq \mathcal{Y}_{\mathrm{m}}(\mathcal{W})$ holds for any user-specified index $l>0$, the inclusion  $\mathcal{Y}_{\mathrm{m}}(\mathcal{W}) \subseteq \mathcal{Y}$ enforced through~\eqref{eq:output_inclusion_O} implies by monotonicity of $\mathrm{d}_{\mathcal{Y}}(\cdot)$ in~\eqref{eq:monotonicty_property} that the inequality	$\mathrm{d}_{\mathcal{Y}}(\mathcal{Y}_{\mathrm{m}}(\mathcal{W})) \leq \mathrm{d}_{\mathcal{Y}}(\mathcal{S}(l,\mathcal{W}))$
holds. Hence, $\mathrm{d}_{\mathcal{Y}}(\mathcal{S}(l,\mathcal{W}))$ upper-bounds $\mathrm{d}_{\mathcal{Y}}(\mathcal{Y}_{\mathrm{m}}(\mathcal{W}))$. 
Thus, we propose to approximate Problem~\eqref{eq:orig_problem_to_solve} to minimize an upper-bound as
\begin{subequations}
\label{eq:approx_problem_to_solve_ideal}
\begin{align}
\min_{\mathcal{W}}\hspace{5pt} & \mathrm{d}_{\mathcal{Y}}(\mathcal{S}(l,\mathcal{W})) 	\label{eq:approx_problem_to_solve_ideal:obj}\\ \text{ s.t. } & \mathcal{O}(s,\alpha,\lambda,\mathcal{W}) \subseteq \mathcal{Y}, 	\label{eq:approx_problem_to_solve_ideal:con_1}
\\ & \0 \in \mathcal{W}. \label{eq:approx_problem_to_solve_ideal:con_2}
\end{align}
\end{subequations}
{\color{black} 
In this problem, the user selects the parameters $H \in \R^{m_B \times n_y}$, $l>0$ and $\mu>0$. We present next a method to compute the parameters $(s,\alpha,\lambda)$ that characterize the RPI set, given some $\mu$. The choice of $l$ is independent of the parameter $s$ that characterizes the $\mu$-RPI set, and a larger value of $l$ can reduce conservativeness of the bound $\mathrm{d}_{\mathcal{Y}}(\mathcal{S}(l,\mathcal{W}))$ to $\mathrm{d}_{\mathcal{Y}}(\mathcal{Y}_{\mathrm{m}}(\mathcal{W}))$.
\begin{remark}
	If the objective function $\mathrm{d}_{\mathcal{Y}}(\cdot)$ is defined as a standard Hausdorff distance induced by some $p$-norm ball $\mathcal{B}_p^{n_y}$ for any $\mathbf{S} \subseteq \mathcal{Y}$ as $	\mathrm{d}_{\mathcal{Y}}(\mathbf{S}):=\min\{ \ \epsilon \ : \ \mathcal{Y} \subseteq \mathbf{S} \oplus \epsilon \mathcal{B}_p^{n_y}\}$,
	then it satisfies the triangle inequality $$\mathrm{d}_{\mathcal{Y}}(\mathcal{Y}_{\mathrm{m}}(\mathcal{W})) \leq \mathrm{d}_{\mathcal{O}(s,\alpha,\lambda,\mathcal{W})}(\mathcal{Y}_{\mathrm{m}}(\mathcal{W}))+\mathrm{d}_{\mathcal{Y}}(\mathcal{O}(s,\alpha,\lambda,\mathcal{W})).$$ Then, we can replace $\mathrm{d}_{\mathcal{Y}}(\mathcal{Y}_{\mathrm{m}}(\mathcal{W}))$ with $\mathrm{d}_{\mathcal{Y}}(\mathcal{O}(s,\alpha,\lambda,\mathcal{W}))$ to minimize an upper-bound. The first term in this bound is characterized by the RPI set approximation error $\mu$ in~\eqref{eq:muRPI_inclusions} as
	\begin{align*}
		\mathrm{d}_{\mathcal{O}(s,\alpha,\lambda,\mathcal{W})}(\mathcal{Y}_{\mathrm{m}}(\mathcal{W})) \leq \min\{ \ \epsilon \ : \ \mu C \mathcal{B}_{{\infty}}^{n_x} \subseteq \epsilon \mathcal{B}_p^{n_y}\}.
	\end{align*} 
	In this paper however, we do not rely on the triangle inequality property of $\mathrm{d}_{\mathcal{Y}}(\cdot)$, thus affording flexibility in its definition. Instead, we use $\mathrm{d}_{\mathcal{Y}}(\mathcal{S}(l,\mathcal{W}))$ that only requires monotonicity of $\mathrm{d}_{\mathcal{Y}}(\cdot)$ to be an upper bound to $\mathrm{d}_{\mathcal{Y}}(\mathcal{Y}_{\mathrm{m}}(\mathcal{W}))$. 
	This choice allows the user to select an appropriate parameter $l>0$ that trades-off between conservativeness and computational complexity. Moreover, it decouples conservativeness in the objective from conservativeness in the feasible region that is characterized by the RPI approximation error $\mu$.  
 	 $\hfill\square$
\end{remark}
}

To ensure that $\mathcal{R}(s,\alpha,\lambda,\mathcal{W})$ is $\mu$-RPI in formulating $\mathcal{O}(s,\alpha,\lambda,\mathcal{W})$ in \eqref{eq:O_outside_Y}, we need to satisfy the inclusions \eqref{eq:RPI_condition_check} and \eqref{eq:distance_condition} at the optimal solution with the parameters $(s,\alpha,\lambda)$. To guarantee this, we can consider the following approaches.
%
%
\begin{enumerate}
\item[{(1)}] \textit{Iterative verification:} Problem~\eqref{eq:approx_problem_to_solve_ideal} is solved with some arbitrary $(s,\alpha,\lambda)$. At the optimizer, inclusions~\eqref{eq:RPI_condition_check} and \eqref{eq:distance_condition} are checked with the chosen $(s,\alpha,\lambda)$. If verified, then $\mathcal{R}(s,\alpha,\lambda,\mathcal{W})$ is a $\mu$-RPI set and the procedure terminates. Else, $(s,\alpha,\lambda)$ are updated at the optimizer to verify~\eqref{eq:RPI_condition_check} and \eqref{eq:distance_condition}, and the procedure is repeated.
\item[{(2)}] \textit{Optimizing over $(s,\alpha,\lambda)$:} The parameters $(s,\alpha,\lambda)$ are appended as optimization variables, and inclusions~\eqref{eq:RPI_condition_check} and \eqref{eq:distance_condition} are appended as constraints in Problem~\eqref{eq:approx_problem_to_solve_ideal}. 
\item[{(3)}] \textit{Preselecting $(s,\alpha,\lambda)$:} Problem~\eqref{eq:approx_problem_to_solve_ideal} is solved with $(s,\alpha,\lambda)$ chosen a priori such that inclusions~\eqref{eq:RPI_condition_check} and \eqref{eq:distance_condition} hold for all feasible disturbance sets $\mathcal{W}$.
\end{enumerate}
Approach~$(1)$ requires an update strategy to update $(s,\alpha,\lambda)$ for convergence guarantees.
Approach~$(2)$, while presenting an attractive alternative since inclusions~\eqref{eq:RPI_condition_check} and \eqref{eq:distance_condition} verify by construction, results in a computationally intractable optimization problem. Further study of these approaches is subject of future research. In this paper, we adopt Approach~$(3)$, i.e., we select parameters $(s,\alpha,\lambda)$ such that inclusions~\eqref{eq:RPI_condition_check} and \eqref{eq:distance_condition} are verified by {\color{black}a set of disturbance sets $\mathcal{W}$ that are feasible for Problem~\eqref{eq:approx_problem_to_solve_ideal}.}
In the next subsection, we present this approach.
\subsection{Preselecting RPI set parameters $(s,\alpha,\lambda)$}
By definition, the requirement 
on parameters $(s,\alpha,\lambda)$ to adopt Approach ($3$) is that inclusions~\eqref{eq:RPI_condition_check} and \eqref{eq:distance_condition} must be verified for all feasible disturbance sets $\mathcal{W}$ of Problem~\eqref{eq:approx_problem_to_solve_ideal}. In this paper, we relax this requirement to the following.
\begin{itemize}
\item
\textit{\textbf{Requirement ($a$)}: For some user-specified $\gamma>0$, parameters $(s,\alpha,\lambda)$ are such that inclusions~\eqref{eq:RPI_condition_check} and \eqref{eq:distance_condition} hold for all disturbance sets
\begin{align}
\label{eq:Lset}
\mathcal{W} \in \bm{\mathcal{L}}_{\gamma}:=\{\mathcal{W} \ : \ \0 \in \mathcal{W}, \ B\mathcal{W} \subseteq \gamma \mathcal{B}_{\infty}^{n_x}\}.
\end{align}}
\end{itemize}
As we will show in the sequel, the restriction $B\mathcal{W} \subseteq \gamma \mathcal{B}_{\infty}^{n_x}$ over the disturbance sets permits us to characterize a set of parameters $(s,\alpha,\lambda)$ satisfying the requirement. The formulation of Requirement ($a$) motivates us to impose the additional constraint $B\mathcal{W} \subseteq \gamma \mathcal{B}_{\infty}^{n_x}$ in Problem~\eqref{eq:approx_problem_to_solve_ideal} to obtain
\begin{subequations}
\label{eq:approx_problem_to_solve_ideal_2}
\begin{align}
\min_{\mathcal{W}}\hspace{5pt} & \mathrm{d}_{\mathcal{Y}}(\mathcal{S}(l,\mathcal{W})) 	\label{eq:approx_problem_to_solve_ideal_2:obj}\\ \text{ s.t. } & \mathcal{O}(s,\alpha,\lambda,\mathcal{W}) \subseteq \mathcal{Y}, 	\label{eq:approx_problem_to_solve_ideal_2:con_1}
\\ & B\mathcal{W} \subseteq \gamma \mathcal{B}_{\infty}^{n_x}, \label{eq:approx_problem_to_solve_ideal_2:con_2}
\\ & \0 \in \mathcal{W}, \label{eq:approx_problem_to_solve_ideal_2:con_3}
\end{align}
\end{subequations}
such that if parameters $(s,\alpha,\lambda)$ satisfy Requirement ($a$), the set $\mathcal{R}(s,\alpha,\lambda,\mathcal{W})$ formulating $\mathcal{O}(s,\alpha,\lambda,\mathcal{W})$ is $\mu$-RPI for every feasible disturbance set $\mathcal{W}$.
%

We impose the following additional requirement on the parameters $(s,\alpha,\lambda)$ to guarantee feasibility of Problem~\eqref{eq:approx_problem_to_solve_ideal_2}.
\begin{itemize}
\item \textit{\textbf{Requirement ($b$)}: Parameters $(s,\alpha,\lambda)$ are such that there exists a disturbance set $\mathcal{W} \in \bm{\mathcal{L}}_{\gamma}$ satisfying constraint~\eqref{eq:approx_problem_to_solve_ideal_2:con_1}.}
\end{itemize}
In the following result, we translate Requirements ($a$) and ($b$) into sufficient conditions on the parameters $(s,\alpha,\lambda)$.
\begin{lemma}
\label{lem:sufficient_parameters}
If the parameters $s>0$, $\alpha \in [0,1)$ and $\lambda \in [0,1]$ verify the inclusion 
\begin{align}
\lambda(1-\alpha)^{-1}\bigoplus_{t=0}^{s-1}CA^t\mathcal{B}_{\infty}^{n_x} \subseteq \mathcal{Y}, \label{eq:inclusion_2_Y_typ}   
\end{align}
then there exists some disturbance set $\mathcal{W} \in \bm{\mathcal{L}}_{\gamma}$ for any $\gamma>0$ such that the inclusion $\mathcal{O}(s,\alpha,\lambda,\mathcal{W}) \subseteq \mathcal{Y}$ holds. Moreover, if for some user-specified scalars $\mu,\gamma>0$, the
inclusions
\begin{align} A^s(\gamma+\lambda)\mathcal{B}_{\infty}^{n_x}&\subseteq \alpha \lambda \mathcal{B}_{\infty}^{n_x}, \label{eq:inclusion_1_typ}  \\ 
(1-\alpha)^{-1}
\bigoplus_{t=0}^{s-1}A^t(\alpha \gamma+\lambda)\mathcal{B}_{\infty}^{n_x} &\subseteq \mu \mathcal{B}_{\infty}^{n_x}, \label{eq:inclusion_2_typ}
\end{align}
hold, then $\mathcal{R}(s,\alpha,\lambda,\mathcal{W})$ is a $\mu$-RPI set corresponding to any disturbance set $\mathcal{W} \in \bm{\mathcal{L}}_{\gamma}$.  	$\hfill\square$
\end{lemma}
\begin{proof}
Note that, by~\eqref{eq:O_outside_Y}, we have
\begin{align*}
\mathcal{O}(s,\alpha,\lambda,\mathcal{W})=(1-\alpha)^{-1} \bigoplus_{t=0}^{s-1}CA^t(B\mathcal{W} \oplus \lambda \mathcal{B}_{\infty}^{n_x}) \oplus D\mathcal{W}.
\end{align*}
Consequently, if inclusion~\eqref{eq:inclusion_2_Y_typ} holds, then the inclusion $\mathcal{O}(s,\alpha,\lambda,\mathcal{W}) \subseteq \mathcal{Y}$ is satisfied by $\mathcal{W}=\{\0\} \in \bm{\mathcal{L}}_{\gamma}$, concluding the proof of the first claim. 
%
%
For the second claim, if inclusion \eqref{eq:inclusion_1_typ} holds, then for any $\mathcal{W} \in \bm{\mathcal{L}}_{\gamma}$, the inclusions
\begin{align*}
A^s(B\mathcal{W} \oplus \lambda \mathcal{B}_{\infty}^{n_x}) &\subseteq A^s(\gamma + \lambda)\mathcal{B}_{\infty}^{n_x}  & \small{\text{$\rightarrow$ Since $B\mathcal{W} \subseteq \gamma \mathcal{B}_{\infty}^{n_x}$}}
\\ & \subseteq \alpha \lambda \mathcal{B}_{\infty}^{n_x}  & \small{\text{$\rightarrow$ Directly from \eqref{eq:inclusion_1_typ}}}
\\ & \subseteq \alpha(B\mathcal{W} \oplus \lambda \mathcal{B}_{\infty}^{n_x})  & \small{\text{$\rightarrow$ Since $\0 \in \alpha B\mathcal{W}$}}
\end{align*} 
%
follow from basic properties of set operations in Proposition~\ref{prop:basic_properties},
such that inclusion \eqref{eq:RPI_condition_check} holds. Similarly, if inclusion~\eqref{eq:inclusion_2_typ} holds, then for every disturbance set $\mathcal{W} \in \bm{\mathcal{L}}_{\gamma}$, the inclusions
\begin{align*}
&(1-\alpha)^{-1}\bigoplus_{t=0}^{s-1}A^t(\alpha B\mathcal{W} \oplus \lambda \mathcal{B}_{\infty}^{n_x})  \\ & \subseteq (1-\alpha)^{-1}\bigoplus_{t=0}^{s-1}A^t(\alpha \gamma+\lambda)\mathcal{B}_{\infty}^{n_x} & \small{\text{$\rightarrow$ Since $B\mathcal{W} \subseteq \gamma \mathcal{B}_{\infty}^{n_x}$}} \nonumber \\
& \subseteq \mu \mathcal{B}_{\infty}^{n_x}  & \small{\text{$\rightarrow$ Directly from \eqref{eq:inclusion_2_typ}}} \nonumber
\end{align*}
follow from basic properties of set operations in Proposition~\ref{prop:basic_properties}, such that inclusion~\eqref{eq:distance_condition} holds.
Consequently, by Lemma~\ref{lemma:original_Rakovic}, 
$\mathcal{R}(s,\alpha,\lambda,\mathcal{W})$ is a $\mu$-RPI set.
\end{proof}

As per Lemma~\ref{lem:sufficient_parameters}, parameters $(s,\alpha,\lambda)$ that verify inclusions~\eqref{eq:inclusion_1_typ}-\eqref{eq:inclusion_2_typ} satisfy Requirement ($a$), and those that verify inclusion~\eqref{eq:inclusion_2_Y_typ} satisfy of Requirement ($b$). Hence, if parameters $(s,\alpha,\lambda)$ verifying inclusions~\eqref{eq:inclusion_2_Y_typ}-\eqref{eq:inclusion_2_typ} are used to formulate Problem~\eqref{eq:approx_problem_to_solve_ideal_2}, then the problem is guaranteed to be feasible, and for every feasible disturbance set $\mathcal{W}$, the set $\mathcal{R}(s,\alpha,\lambda,\mathcal{W})$ formulating $\mathcal{O}(s,\alpha,\lambda,\mathcal{W})$ is $\mu$-RPI. 
%
%
%
%
%
%
%

\subsection{Computing RPI set parameters $(s,\alpha,\lambda)$}
We will now present an algorithm to compute parameters $(s,\alpha,\lambda)$ verifying inclusions~\eqref{eq:inclusion_2_Y_typ}-\eqref{eq:inclusion_2_typ}. To this end, we define the set of all parameters $(s,\alpha,\lambda)$ satisfying these inclusions for some user-specified scalars $\gamma,\mu>0$ as
\begin{align*}
{\Lambda}_{\gamma,\mu}:=\left\{(s,\alpha,\lambda) : s>0, \alpha \in [0,1),\lambda \in [0,1], \ \eqref{eq:inclusion_2_Y_typ}-\eqref{eq:inclusion_2_typ}\right\}.
\end{align*}
Then, suitable parameters $(s,\alpha,\lambda)$ can be selected by solving 
\begin{align}
\label{eq:sal_selection_problem}
\min_{s,\alpha,\lambda} \quad s \quad \text{s.t.} \quad (s,\alpha,\lambda) \in {\Lambda}_{\gamma,\mu}.
\end{align}
In Problem~\eqref{eq:sal_selection_problem}, we compute the smallest index $s$ satisfying inclusions~\eqref{eq:inclusion_2_Y_typ}-\eqref{eq:inclusion_2_typ}, since the index $s$ characterizes the number of Minkowski sums defining the set $\mathcal{O}(s,\alpha,\lambda,\mathcal{W})$ in~\eqref{eq:approx_problem_to_solve_ideal_2:con_1}, and a smaller number is desirable for reduced computational complexity.
Unfortunately, solving  Problem~\eqref{eq:sal_selection_problem} directly is not viable since $s$ is a discrete variable. Thus, we propose to solve Problem~\eqref{eq:sal_selection_problem} iteratively by incrementing $s$ and searching for feasible parameters $(\alpha,\lambda)$.

For some $s>0$, we define the set of parameters $(\alpha,\lambda)$ as
\begin{align}
\label{eq:L_HF_set}
\bm{\mathcal{L}}_{\gamma,\mu}(s):=\{(\alpha,\lambda) \ : \ (s,\alpha,\lambda) \in {\Lambda}_{\gamma,\mu}\},
\end{align}
based on which we define the optimization problem
\begin{align}
\mathbb{H}(s) \ 
\begin{cases} \quad
\underset{\alpha,\lambda}{\max} \ \ \alpha+\lambda 
\ \ \text{s.t.} \ \  (\alpha,\lambda) \in \bm{\mathcal{L}}_{\gamma,\mu}(s).
\end{cases}
\end{align}

Using this problem, we define the iterative procedure to solve Problem~\eqref{eq:sal_selection_problem} in Algorithm~\ref{alg:SAL_solving_procedure}. 
\begin{algorithm}[t]
\caption{Solving Problem~\eqref{eq:sal_selection_problem}}
\begin{algorithmic}[1]
\State \textbf{Require} Matrices $A,C,G$, vector $g$, scalars $\gamma,\mu >0$
\State \textbf{Initialize:} $s=1, \ conv=0$
\While{$conv=0$}
\State Solve $\mathbb{H}(s)$ for $(\alpha,\lambda)$;
\If{$\mathbb{H}(s)$ is infeasible}
\State $s \leftarrow s+1$ 
\Else
\State $conv \leftarrow 1$
\EndIf
\EndWhile
\State \textbf{Return:} $(s,\alpha,\lambda)$
\end{algorithmic}
\label{alg:SAL_solving_procedure}
\end{algorithm}
In the formulation of $\mathbb{H}(s)$, we maximize $\alpha+\lambda$ for numerical stability. However, since we aim for feasibility, any objective can be used. We now show that $\mathbb{H}(s)$ can be implemented as a Second-Order Cone Program (SOCP)~\cite[Sec 4.4.2]{boyd2004convex}. Note that since $\alpha \in [0,1)$ and $\lambda \in [0,1]$, $\mathbb{H}(s)$ is bounded if feasible. 

\subsubsection{Implementation of $\mathbb{H}(s)$}
To implement $\mathbb{H}(s)$, we encode inclusions~\eqref{eq:inclusion_2_Y_typ}-\eqref{eq:inclusion_2_typ} using support functions.
Recalling that the set $\mathcal{Y}=\{y:Gy \leq g\}$ with $G \in \R^{m_{\mathcal{Y}} \times n_y}$ from~\eqref{eq:Y_orig_defn}, and denoting  $\tilde{\I}_{n_x}:=[\I_{n_x} \ \ -\I_{n_x}]^{\top},$
we note from~\eqref{eq:SF_prop_1}-\eqref{eq:SF_prop_2} that
\begin{subequations}
\begin{align}
\eqref{eq:inclusion_2_Y_typ} &\iff&	\frac{\lambda}{1-\alpha} \sum_{t=0}^{s-1} h_{CA^t \mathcal{B}_{\infty}^{n_x}}(G) &\leq g, \vspace{3pt} \label{eq:SF_2_Y_typ}\\
\eqref{eq:inclusion_1_typ} &\iff&	(\gamma+\lambda) h_{A^s \mathcal{B}_{\infty}^{n_x}}(\tilde{\I}_{n_x}) &\leq \alpha \lambda \1, \vspace{3pt} \label{eq:SF_1_typ} \\
\eqref{eq:inclusion_2_typ} &\iff&	\frac{\alpha \gamma + \lambda}{1-\alpha} \sum_{t=0}^{s-1} h_{A^t \mathcal{B}^{n_x}_{\infty}}(\tilde{\I}_{n_x}) &\leq \mu \1.  \label{eq:SF_2_typ}
\end{align}
\end{subequations}
Exploiting Equation~\eqref{eq:box_SF}, we then define the constants
\begin{subequations}
\label{eq:MLs_defn}
\begin{align}
L^{[s]}&:=\sum_{t=0}^{s-1}h_{CA^t\mathcal{B}_{\infty}^{n_x}}(G)=\sum_{t=0}^{s-1}|GCA^t|\1, \label{eq:Ls_definition} \\
\theta^{[s]}&:=\min_{i \in \mathbb{I}_1^{m_{\mathcal{Y}}}} \left\{g_i/L^{[s]}_i\right\},  \label{eq:thetas_definition} \\
M^{[s]}&:=\norm{\sum_{t=0}^{s-1}h_{A^t\mathcal{B}_{\infty}^{n_x}}(\tilde{\I}_{n_x})}_{\infty}=\norm{\sum_{t=0}^{s-1}|\tilde{\I}_{n_x}A^t|\1}_{\infty}, \label{eq:Ms_definition}
\end{align}	
\end{subequations}
and observe that the support function $$h_{A^s \mathcal{B}^{n_x}_{\infty}}(\tilde{\I}_{n_x})=|\tilde{\I}_{n_x} A^s|\1 \leq \norm{A^s}_{\infty}\1,$$ following the usual definition of $\infty$-norm for matrices.
Note that $L^{[s]} \in \R^{m_{\mathcal{Y}}}$ following from dimensions of matrix $G$.
Then, the support function inequalities~\eqref{eq:SF_2_Y_typ}-\eqref{eq:SF_2_typ} can be written after simple algebraic manipulations as
\begin{subequations}
\begin{align}
\eqref{eq:SF_2_Y_typ} &&\iff&& \lambda \leq (1-\alpha) \theta^{[s]}, \label{eq:ineq_2_Y_typ} \\
\eqref{eq:SF_1_typ} &&\iff&& (\gamma+\lambda) \norm{A^s}_{\infty} \leq \alpha \lambda, \label{eq:ineq_1_typ}\\
\eqref{eq:SF_2_typ} &&\iff&& (\alpha \gamma + \lambda) M^{[s]} \leq (1-\alpha) \mu. \label{eq:ineq_2_typ}
\end{align}
\end{subequations} 
Hence, for a given $s>0$, $\mathbb{H}(s)$ can be written as
\begin{subequations}
\label{eq:Hs_prop_1}
\begin{align}
\max_{\alpha,\lambda} & \hspace{5pt} \alpha+\lambda \\
\text{s.t.} & \hspace{5pt} \eqref{eq:ineq_2_Y_typ}, \  \eqref{eq:ineq_1_typ}, \ \eqref{eq:ineq_2_typ}, \\
& \hspace{5pt} \alpha \in [0,1), \ \lambda \in [0,1].
\end{align}
\end{subequations}
While Constraints~\eqref{eq:SF_2_Y_typ} and~\eqref{eq:SF_2_typ} in Problem~\eqref{eq:Hs_prop_1} are linear in $(\alpha,\lambda)$, Constraint~\eqref{eq:ineq_1_typ} is nonlinear. However, it can be written as a second-order cone (SOC) by exploiting the fact that $\alpha,\lambda \geq 0$  in the feasible domain as follows. In these arguments, we denote $\zeta:=\norm{A^s}_{\infty}$  for notational convenience.
\vspace{-15pt}
{\color{black}
\begin{subequations}
\begin{align}
&(\gamma+\lambda) \zeta \leq \alpha \lambda \quad
\Leftrightarrow  \quad \gamma\zeta \leq (\alpha-\zeta) \lambda \label{eq:SOC_constraint_pos}\\
&\Leftrightarrow  4\gamma\zeta+(\alpha -\zeta-\lambda)^2 \leq (\alpha -\zeta+\lambda)^2 \label{eq:SOC_constraint_sqrt_bef}\\
&\Leftrightarrow   \norm{ \begin{bmatrix} 2 \sqrt{\gamma\zeta} \\ \alpha -\zeta-\lambda \end{bmatrix} }_2 \leq \alpha -\zeta+\lambda. \label{eq:SOC_constraint}
\end{align}
\end{subequations}}
%
Taking the square-root on both sides of the inequality in~\eqref{eq:SOC_constraint_sqrt_bef} is valid since $\alpha-\zeta \geq 0$ from~\eqref{eq:SOC_constraint_pos} such that $\alpha-\zeta+\lambda \geq 0$, and the left-hand-side of~\eqref{eq:SOC_constraint_sqrt_bef} is nonnegative. 
The inequality in~\eqref{eq:SOC_constraint} is an SOC, such that Problem~\eqref{eq:Hs_prop_1} is an SOCP in two-dimensions that can be solved very efficiently.

\subsubsection{Termination of Algorithm~\ref{alg:SAL_solving_procedure}}
Algorithm~\ref{alg:SAL_solving_procedure} terminates in finite time 
if and only if, for the user-specified $\gamma,\mu>0$, there exists some index $s>0$ such that $\mathbb{H}(s)$ is feasible. This is equivalent to the existence of some $s>0$ such that the domain $\bm{\mathcal{L}}_{\gamma,\mu}(s)$ of $\mathbb{H}(s)$ defined in~\eqref{eq:L_HF_set} is nonempty, that is in turn equivalent to non-emptiness of the set $\Lambda_{\gamma,\mu}$. 
{\color{black} In the following result, we show that indeed $\Lambda_{\gamma,\mu}$ is nonempty by following an approach similar to that in~\cite{Rakovic2006}. 
 However, our method differs from this previous work due to the additional requirement of satisfying inclusions~\eqref{eq:inclusion_2_Y_typ} and~\eqref{eq:inclusion_1_typ}.}
%
%
%
\begin{theorem}
\label{thm:formulation_lem}
Suppose that
Assumptions~\ref{ass:0_in_Y_set} and \ref{ass:stable} hold. Then the set $\Lambda_{\gamma,\mu}$ of parameters $(s,\alpha,\lambda)$ verifying~\eqref{eq:inclusion_2_Y_typ}-\eqref{eq:inclusion_2_typ}  is nonempty for any user-specified $\gamma>0$ and $\mu>0$. $\hfill\square$
\end{theorem}
\begin{proof}
For any given $\gamma,\mu>0$ and some $s>0$, we recall from~\eqref{eq:ineq_2_typ} that inclusion~\eqref{eq:inclusion_2_typ} is verified if and only if
\begin{align}
\label{eq:to_check_in_Algo_1}
(1-\alpha)^{-1}(\alpha \gamma+\lambda)M^{[s]} \leq \mu.
\end{align}
By rearranging this inequality, we obtain
\begin{align}
\label{eq:to_check_in_Algo_1:1}
\alpha \leq  (\mu - \lambda M^{[s]})/(\mu + \gamma M^{[s]}).
\end{align}
Recalling the definition of $M^{[s]}$ from~\eqref{eq:Ms_definition}, we define its limit $\hat{M}:=\underset{s \to \infty}{\lim} M^{[s]}$,
and observe that since $\bigoplus_{t=0}^{\infty} A^t \mathcal{B}^{n_x}_{\infty}$ is compact under Assumption~\ref{ass:stable} \cite{Kolmanovsky1998} and the inclusion 
\begin{align}
\label{eq:truncated_sum}
\bigoplus_{t=0}^{s-1} A^t \mathcal{B}^{n_x}_{\infty} \subseteq \bigoplus_{t=0}^{\infty} A^t \mathcal{B}^{n_x}_{\infty}
\end{align}
holds for any $s > 0$, the inequalities $M^{[s]} \leq \hat{M} < \infty$ hold because of Equations~\eqref{eq:SF_prop_3} and \eqref{eq:SF_prop_2}. This implies that 
$$(\mu - \lambda \hat{M})/(\mu + \gamma \hat{M}) \leq  (\mu - \lambda M^{[s]})/(\mu + \gamma M^{[s]})$$
holds for any $s>0$. Hence, for some user-specified $\gamma,\mu>0$, if we select parameters $\alpha,\lambda \in (0,1)$ verifying the inequality
\begin{align}
\label{eq:cond_1_to_satisfy}
\alpha \leq (\mu - \lambda \hat{M})/(\mu + \gamma \hat{M}),
\end{align}
then inclusion~\eqref{eq:inclusion_2_typ} will be verified for any $s>0$.  

Regarding inclusion~\eqref{eq:inclusion_2_Y_typ}, we recall from~\eqref{eq:ineq_2_Y_typ} that it holds if and only if
\begin{align}
(1-\alpha)^{-1}{\lambda} \leq \theta^{[s]} \iff \alpha \leq 1-\lambda/\theta^{[s]}. \label{eq:to_check_in_Algo_2:2}
\end{align}
Recalling the definition of $\theta^{[s]}$ from~\eqref{eq:thetas_definition}, we define its limit
\begin{align}
\label{eq:thetahat_defns}
\hat{\theta}:=\lim_{s \to \infty} \theta^{[s]}.
\end{align}
From the definition of $L^{[s]}$ in~\eqref{eq:Ls_definition}, we observe that
$L^{[s]}_i$ is monotonically nondecreasing in $s$ for each component $i \in \mathbb{I}_1^{m_{\mathcal{Y}}}$ and $g$ is a constant, $\theta^{[s]}$ is monotonically nonincreasing in $s$. We then define $\hat{L}_i:=\underset{s \to \infty}{\lim} L_i^{[s]}$ such that $\hat{\theta}=\underset{{i \in \mathbb{I}_1^{m_{\mathcal{Y}}}}}{\min} \{g_i/\hat{L}_i\}$.
Since $\hat{L}_i<\infty$ under Assumption~\ref{ass:stable} and $g_i>0$ under Assumption~\ref{ass:0_in_Y_set} for each component
$i \in \mathbb{I}_1^{m_{\mathcal{Y}}}$, the inequalities $\theta^{[s]} \geq \hat{\theta}> 0$ hold. This implies that
$1- \lambda/\hat{\theta} \leq 1-\lambda/\theta^{[s]},$
such that that if we select some $\alpha,\lambda \in (0,1)$ satisfying
\begin{align}
\label{eq:cond_2_to_satisfy}
\alpha \leq 1-\lambda/\hat{\theta},
\end{align}
then inclusion~\eqref{eq:inclusion_2_Y_typ} will be verified for any $s>0$. 

Thus, if there exist parameters $\alpha,\lambda \in (0,1)$ verifying inequalities~\eqref{eq:cond_1_to_satisfy} and \eqref{eq:cond_2_to_satisfy}, then these parameters verify inclusions~\eqref{eq:inclusion_2_typ} and~\eqref{eq:inclusion_2_Y_typ} for all $s>0$. We now demonstrate the existence of such parameters. 
To this end, we define
\begin{align}
\label{eq:qhat_defn}
\hat{q}:=\min\left\{\mathrm{min}\{\mu,\hat{M}\}/\hat{M} \ , \quad \hat{\theta}\right\},
\end{align}
and select $\lambda=\delta \hat{q}$ for some $\delta \in (0,1)$.
For any user-specified $\mu>0$, we have $\hat{q} \in (0,1)$ since $\hat{\theta},\hat{M}>0$. This implies that $\lambda=\delta \hat{q} \in (0,1)$ for any $\delta \in (0,1)$. We also define
\begin{align}
\label{eq:rhat_defn}
\hat{r}(\delta):=\min\left\{(\mu - \delta \hat{q}\hat{M})/(\mu + \gamma \hat{M}) \ , \quad 1-\delta \hat{q}/\hat{\theta}\right\}.
\end{align}
Then, substituting $\lambda=\delta \hat{q}$ in inequalities~\eqref{eq:cond_1_to_satisfy} and \eqref{eq:cond_2_to_satisfy}, we observe that if $\hat{r}(\delta) \in (0,1)$, then any $\alpha \in(0,\hat{r}(\delta)]$	verifies these inequalities. Hence, it remains to show that $\hat{r}(\delta) \in (0,1)$. 

To that end, we note that since the inequalities  $\delta\hat{q}/\hat{\theta}>0$ and $\mu - \delta \hat{q}\hat{M} < \mu + \gamma \hat{M}$
hold for every $\delta \in (0,1)$, we always have $\hat{r}(\delta)<1$. In order to show that $\hat{r}(\delta)>0$, we consider the following two cases, which are based on Equation~\eqref{eq:qhat_defn}. 

\underline{Case 1:} if $\hat{q}=\mu/\hat{M} \leq \hat{\theta},$ for any $\delta \in (0,1)$,
we have either $\hat{r}(\delta)=\mu(1-\delta)/(\mu + \gamma \hat{M})$ or $\hat{r}(\delta)=1-\delta\mu/(\hat{M}\hat{\theta})$ .
Since $\mu,\gamma,\hat{M}>0$, the first option satisfies $\hat{r}(\delta)>0$. Regarding the second option, the inequality $\mu/\hat{M} \leq \hat{\theta}$ implies
\begin{align*}
\delta\mu<\mu \leq \hat{M}\hat{\theta} && \Rightarrow && 1-\delta\mu/(\hat{M}\hat{\theta})=\hat{r}(\delta)>0.
\end{align*}

%

\underline{Case 2:} if $\hat{q}=\hat{\theta} \leq \mu/\hat{M},$  for any $\delta \in (0,1)$ we have
either $\hat{r}(\delta)=(\mu-\delta \hat{\theta}\hat{M})/(\mu + \gamma \hat{M})$ or $\hat{r}(\delta)=1-\delta$.
Since $\delta \in (0,1)$, the second option always satisfies $\hat{r}(\delta)>0$. Regarding the first option, the inequality $\hat{\theta} \leq \mu/\hat{M}$ implies
\begin{align*}
\delta \hat{M}\hat{\theta} < \hat{M}\hat{\theta} \leq \mu && \Rightarrow && \mu-\delta \hat{M}\hat{\theta}>0,
\end{align*}
such that $\hat{r}(\delta)>0$ since $\mu,\gamma,\hat{M}>0$. Thus, for any $\delta \in (0,1)$, we have $\hat{r}(\delta) \in (0,1)$, such that 
the parameters $\lambda=\delta \hat{q}$ and $\alpha \in (0,\hat{r}(\delta)]$
verify inclusions~\eqref{eq:inclusion_2_typ} and~\eqref{eq:inclusion_2_Y_typ} for all $s>0$.

\noindent
Now, we show that there exists some $\alpha \in (0,\hat{r}(\delta)]$ and $s>0$ also satisfying inclusion \eqref{eq:inclusion_1_typ} with $\lambda=\delta \hat{q}$. To this end, we recall from~\eqref{eq:ineq_1_typ} that inclusion~\eqref{eq:inclusion_1_typ} is verified if and only if
\begin{align}
\label{eq:ineq_to_hold_3}
(\gamma+\lambda) \norm{A^s}_{\infty} \leq \alpha \lambda
\end{align}
holds. Then, defining
\begin{align}
\label{eq:alpha_s_definition}
\alpha_{[s]}(\lambda):=\left(1+\gamma/\lambda\right)\norm{A^s}_{\infty},
\end{align}
we observe that for a given $\lambda \in (0,1)$, $\gamma>0$ and $s>0$, $\alpha_{[s]}(\lambda)$ is the smallest value of $\alpha$ verifying inequality~\eqref{eq:ineq_to_hold_3}. This implies that if $\lambda=\delta \hat{q}$, then $\alpha_{[s]}(\delta \hat{q})$ verifies~\eqref{eq:ineq_to_hold_3}.
Then, since Assumption~\ref{ass:stable} entails that~\cite{Hogb06} 
$\underset{s \to \infty}{\lim} \norm{A^s}_{\infty} = 0,$ 
there always exists some $s=\hat{s}(\delta)$ such that $\alpha_{[\hat{s}(\delta)]}(\delta \hat{q})\leq \hat{r}(\delta)$. 
Since such a triplet of parameters 
\begin{align}
\label{eq:parameters_feasible}
(s=\hat{s}(\delta),\alpha\in (0,\hat{r}(\delta)],\lambda=\delta \hat{q}), && \forall \ \delta \in (0,1),
\end{align}
also satisfies \eqref{eq:inclusion_2_typ} and \eqref{eq:inclusion_2_Y_typ}, the set $\Lambda_{\gamma,\mu}$ is nonempty.  
\end{proof}

Thus, using Algorithm~\ref{alg:SAL_solving_procedure}, feasible parameters $(s,\alpha,\lambda)$ can be computed, using which Problem~\eqref{eq:approx_problem_to_solve_ideal_2} can be formulated with the guarantee that the set $\mathcal{R}(s,\alpha,\lambda,\mathcal{W})$ is $\mu$-RPI for every feasible disturbance set $\mathcal{W}$. We now recall again the formulation of this problem, in which we explicitly incorporate the objective $\mathrm{d}_{\mathcal{Y}}(\cdot)$ defined in Equation~\eqref{eq:objective_exact}:
\begin{subequations}
\label{eq:final_formulation_to_solve}
\begin{align}
\hspace{-5pt}
&\min_{\epsilon,\mathcal{W}} \quad \norm{\epsilon}_1 \\ & \quad \text{s.t.}  \quad
\mathcal{O}(s,\alpha,\lambda,\mathcal{W})\subseteq \mathcal{Y}, \label{eq:final_formulation_to_solve:con1} \\
&\qquad \quad \	B\mathcal{W} \subseteq \gamma \mathcal{B}^{n_x}_{\infty}, \label{eq:final_formulation_to_solve:con2} \\
&\qquad \quad \	\0 \in \mathcal{W}, \label{eq:final_formulation_to_solve:con3}  \\
&\qquad \quad \	\mathcal{Y} \subseteq \mathcal{S}(l,\mathcal{W}) \oplus \mathbb{B}(\epsilon). \label{eq:final_formulation_to_solve:con4} 
\end{align}
\end{subequations}
In the sequel, we present an efficient encoding of Problem~\eqref{eq:final_formulation_to_solve} followed by an optimization algorithm.

\begin{remark}
\label{remark:gamma_condition}
The conservativeness introduced due to constraint~\eqref{eq:final_formulation_to_solve:con2}, i.e., $B\mathcal{W} \subseteq \gamma \mathcal{B}_{\infty}^{n_x}$, in Problem~\eqref{eq:final_formulation_to_solve} with respect to Problem~\eqref{eq:approx_problem_to_solve_ideal}, can be eliminated by selecting a $\gamma>0$ large enough such that every disturbance set $\mathcal{W}$ feasible for constraint~\eqref{eq:final_formulation_to_solve:con1} satisfies~\eqref{eq:final_formulation_to_solve:con2}, thus rendering~\eqref{eq:final_formulation_to_solve:con2} inactive and making Problem~\eqref{eq:final_formulation_to_solve} fully equivalent to Problem~\eqref{eq:approx_problem_to_solve_ideal}. This, however, increases the complexity of Problem~\eqref{eq:final_formulation_to_solve}. As $\gamma>0$ increases for some fixed $\lambda \in (0,1)$ and $s>0$, the smallest value of $\alpha=\alpha_{[s]}(\lambda)$ verifying~\eqref{eq:ineq_to_hold_3} increases, as observed from the definition of~$\alpha_{[s]}(\lambda)$ in Equation~\eqref{eq:alpha_s_definition}. Then, the value $s>0$ required to verify~\eqref{eq:to_check_in_Algo_1:1} and~\eqref{eq:to_check_in_Algo_2:2} with $\alpha=\alpha_{[s]}(\lambda)$ increases. This implies that the set $\mathcal{O}(s,\alpha,\lambda,\mathcal{W})$ in~\eqref{eq:O_outside_Y} and formulating~\eqref{eq:final_formulation_to_solve:con1} requires a larger number of Minkowski sums for its characterization, thus increasing complexity.
$\hfill\square$
\end{remark}
{\color{black}
\begin{remark}
\label{remark:slow_system}
The level of conservativeness introduced by the parameter $\mu > 0$ is dependent on the scale of the set $\mathcal{Y}$ and singular values of matrix $C$, as observed from~\eqref{eq:S1S2S3_defn}. Clearly, smaller values of $\mu$ lead to reduced conservativeness. However, for systems with $\rho(A)\approx 1$, small values of $\mu$ require large values of parameter $s$ to verify~\eqref{eq:to_check_in_Algo_1:1} and~\eqref{eq:to_check_in_Algo_2:2}. This, in turn, increases the complexity of Problem~\eqref{eq:final_formulation_to_solve}, as more Minkowski sums are needed to characterize $\mathcal{O}(s,\alpha,\lambda,\mathcal{W})$. Thus, $\mu$ must be selected by considering these three aspects.
%
	$\hfill\square$
\end{remark}
}
\section{Solving Problem \eqref{eq:final_formulation_to_solve}}
\label{sec:optimization_algorithm}
We now present a tractable encoding of the constraints of Problem \eqref{eq:final_formulation_to_solve}, followed by an optimization algorithm.
\subsection{Inclusions \eqref{eq:final_formulation_to_solve:con1} and \eqref{eq:final_formulation_to_solve:con2}} 
To encode inclusion~\eqref{eq:final_formulation_to_solve:con1}, i.e., $\mathcal{O}(s,\alpha,\lambda,\mathcal{W})\subseteq \mathcal{Y}$, 
we recall that the constraint set is given by $\mathcal{Y}=\{y:Gy \leq g\},$ and the set $\mathcal{O}(s,\alpha,\lambda,\mathcal{W})$ is defined in Equation~\eqref{eq:O_outside_Y} as
\begin{align*}
\mathcal{O}(s,\alpha,\lambda,\mathcal{W})=(1-\alpha)^{-1}\bigoplus_{t=0}^{s-1}CA^t(B\mathcal{W} \oplus \lambda \mathcal{B}_{\infty}^{n_x}) \oplus D\mathcal{W}.
\end{align*}
Defining the matrices
\begin{align}
\label{eq:Gt_definitions}
\bar{G}_{[t]}:=(1-\alpha)^{-1}GCA^t, && \forall \ t \in \mathbb{I}_0^{s-1},
\end{align}
we observe that the inclusion is verified according to Equations~\eqref{eq:SF_prop_1}-\eqref{eq:SF_prop_2} if and only if the support function inequality
\begin{align}
\label{eq:con_1_SF}
\sum_{t=0}^{s-1}h_{B\mathcal{W}}(\bar{G}_{[t]})+h_{D\mathcal{W}}(G)\leq g-\lambda\sum_{t=0}^{s-1}h_{ \mathcal{B}_{\infty}^{n_x}}(\bar{G}_{[t]})
\end{align}
%
holds .
Similarly, inclusion \eqref{eq:final_formulation_to_solve:con2}, i.e., $B\mathcal{W} \subseteq \gamma \mathcal{B}^{n_x}_{\infty}$, is verified if and only if the support function inequality
\begin{align}
h_{B\mathcal{W}}(\tilde{\I}_{n_x}) &\leq \gamma\1, \label{eq:con_2_SF}
\end{align}
holds according to Equation~\eqref{eq:SF_prop_2}.
In order to encode the support function inequalities in~\eqref{eq:con_1_SF} and \eqref{eq:con_2_SF} efficiently, we propose to use the disturbance set parameterization 	
\begin{subequations}
\label{eq:convex_hull_of_boxes}
\begin{align}
&\hspace{12pt} \mathcal{W}=\mathrm{ConvHull}\left(\mathbb{W}(\bar{w}_{[j]},\epsilon^w_{[j]}), \ j \in \mathbb{I}_1^N\right), \\
&\mathbb{W}(\bar{w}_{[j]},\epsilon^w_{[j]}):=\bar{w}_{[j]} \oplus \{w: -\epsilon^w_{[j]} \leq w \leq \epsilon^w_{[j]}\}, 
\end{align}
\end{subequations}
i.e., as a convex hull of boxes $\{\mathbb{W}(\bar{w}_{[j]},\epsilon^w_{[j]}), j \in \mathbb{I}_1^N\}$ where $N \geq 1$ is a user-specified amount of boxes, such that $\mathcal{W}$ is characterized by parameters $\{\bar{w}_{[j]},\epsilon^w_{[j]} \in \R^{n_w}, j \in \mathbb{I}_1^N\}$. This parameterization presents a representational advantage over popular polytopic parameterizations such as zonotopes \cite{Sadraddini2019} that are constrained to be symmetric and hence can be conservative, and a computational advantage over parameterizations such as halfspace-representations~\cite{Mulagaleti2020}, zonotopic intersections~\cite{Althoff2011}, constrained zonotopes~\cite{Scott2016}, etc. that are not immediately amenable to a simple encodings of the support function inequalities~\eqref{eq:con_1_SF}-\eqref{eq:con_2_SF}. {\color{black} The number of variables required to represent this parameterization is $N \times 2n_w$. We note that a pure vertex parameterization of $\mathcal{W}$ is a special case of the parameterization in~\eqref{eq:convex_hull_of_boxes} with the additional constraint $\epsilon^w_{[j]}=0$ for all $j \in \mathbb{I}_1^N$. Hence, the techniques presented in the sequel apply to the vertex parameterization with minor modifications. 
\begin{remark}
In order to represent a polytope $\tilde{\mathcal{W}}$ characterized by $M$ vertices, a purely vertex parameterization of $\mathcal{W}$ requires $L_{\mathrm{v}}=M \times n_w$ number of variables, while the parameterization in~\eqref{eq:convex_hull_of_boxes} requires $L_{\mathrm{b}} \in \left[\mathrm{ceil}\left({M}/{2^{n_w}}\right) , M \right]\times 2n_w$ number of variables. The lower bound on $L_{\mathrm{b}}$ follows since each box can describe at most $2^{n_w}$ vertices of $\tilde{\mathcal{W}}$, and the upper bound follows if each $\epsilon^w_{[j]}=\0$. The exact value of $L_{\mathrm{b}}$ depends on the geometry of  $\tilde{\mathcal{W}}$, and the development of results relating them is a subject of future study. In our experiments on high-dimensional systems, we frequently encountered $L_{\mathrm{b}} << L_{\mathrm{v}}$, emperically demonstrating that the parameterization in~\eqref{eq:convex_hull_of_boxes} would be preferable over a pure vertex parameterization of $\mathcal{W}$. This is expected, since it is known that the number of vertices $M$ can increase exponentially with the dimension $n_w$.  $\hfill\square$
\end{remark}	} 
%
%
%

To encode~\eqref{eq:con_1_SF} and \eqref{eq:con_2_SF} for the disturbance set parameterization in~\eqref{eq:convex_hull_of_boxes}, we rely on the following general result regarding support functions over convex hulls of polytopes.
\begin{proposition}
\label{prop:conv_hull}
Given any polytopes $\{\mathcal{Q}_j \subset \R^n, \ j \in \mathbb{I}_1^q\}$, and denoting the convex hull $\hat{\mathcal{Q}}:=\mathrm{ConvHull}(\mathcal{Q}_j, \ j \in \mathbb{I}_1^q)$, then for any matrix $\mathbf{T} \in \R^{l \times n}$ and vector $\mathbf{p} \in \R^{l}$, the support function for the set $\mathbf{T}\hat{\mathcal{Q}}$ at $\mathbf{p}$ is given by \begin{align}
\label{eq:convhull_SF_toshow}
\hspace{60pt}
h_{\mathbf{T}\hat{\mathcal{Q}}}(\mathbf{p})=\max_{j \in \mathbb{I}_1^q} \ h_{\mathbf{T}\mathcal{Q}_j}(\mathbf{p}). && \qquad \quad   \hfill\square
\\ \nonumber
\end{align}
\end{proposition}
\begin{proof}
Since the set $\hat{\mathcal{Q}}$ is a polytope that is the convex hull of polytopes $\{\mathcal{Q}_j, \ j \in \mathbb{I}_1^q\}$, we know that 
\begin{subequations}
\begin{align}
\mathrm{vert}(\hat{\mathcal{Q}}) \ &\subseteq \ \cup_{j=1}^q \mathrm{vert}(\mathcal{Q}_j), \label{eq:contradiction_1} \\
\nexists \ \tilde{z} \ : \ \tilde{z} &\in \cup_{j=1}^q  \mathcal{Q}_j, \  \tilde{z} \notin \hat{\mathcal{Q}}. \label{eq:contradiction_2}
\end{align}
\end{subequations}
Defining $\mathbf{r}:=\mathbf{T}^{\top} \mathbf{p}$, we know from the support function definition in \eqref{eq:SF_prop_1} that \eqref{eq:convhull_SF_toshow} holds if and only if
\begin{align}
\label{eq:convhull_LP_toshow}
\max_{z \in \hat{\mathcal{Q}}} \ \mathbf{r}^{\top} z = \max_{j \in \mathbb{I}_1^q} \ \left\{\max_{z_{[j]} \in \mathcal{Q}_j} \ \mathbf{r}^{\top} z_{[j]} \right\}.
\end{align}
We now prove \eqref{eq:convhull_LP_toshow} by contradiction.
Suppose that
\begin{align}
\label{eq:convhull_LP_toshow_contra1}
\max_{z \in \hat{\mathcal{Q}}} \ \mathbf{r}^{\top} z \ > \ \max_{j \in \mathbb{I}_1^q} \ \left\{\max_{z_{[j]} \in \mathcal{Q}_j} \ \mathbf{r}^{\top} z_{[j]} \right\},
\end{align}
which holds if and only if there exists a vertex ${z^* \in \mathrm{vert}(\hat{\mathcal{Q}})}$ such that $\mathbf{r}^{\top} z^*=\underset{z \in \hat{\mathcal{Q}}}{\max} \ \mathbf{r}^{\top} z \ \ \text{and} \ \ z^* \notin\cup_{j=1}^q \mathrm{vert}(\mathcal{Q}_j).$ Since this contradicts \eqref{eq:contradiction_1}, \eqref{eq:convhull_LP_toshow_contra1} cannot hold.
Now, suppose
\begin{align}
\label{eq:convhull_LP_toshow_contra2}
\max_{z \in \hat{\mathcal{Q}}} \ \mathbf{r}^{\top} z \ < \ \max_{j \in \mathbb{I}_1^q} \ \left\{\max_{z_{[j]} \in \mathcal{Q}_j} \ \mathbf{r}^{\top} z_{[j]} \right\}, 
\end{align}
which holds if and only if there exists some $\tilde{z}^*$ such that ${\tilde{z}^* \in\cup_{j=1}^q  \mathcal{Q}_j} \ \ \text{and}  \ \ \tilde{z}^* \notin \hat{\mathcal{Q}}.$ Since this contradicts \eqref{eq:contradiction_2}, \eqref{eq:convhull_LP_toshow_contra2} cannot hold and the proof is complete.
\end{proof}

Using Proposition~\ref{prop:conv_hull} with $\mathcal{Q}_j = \mathbb{W}(\bar{w}_{[j]},\epsilon^w_{[j]})$, $q=N$ and $\hat{\mathcal{Q}}=\mathcal{W}$ as per the definition of the disturbance set in~\eqref{eq:convex_hull_of_boxes}, and recalling the expression for support functions of over boxes from~\eqref{eq:box_SF}, the support function over $\mathcal{W}$ for any given matrix $\mathbf{T} \in \R^{l \times n_w}$ and vector $\mathbf{p} \in \R^{l}$ is obtained as
\begin{align}
\label{eq:box_SF_hull}
h_{\mathbf{T}\mathcal{W}}(\mathbf{p})=\max_{j \in \mathbb{I}_1^N} \ \{ \mathbf{p}^{\top} \mathbf{T}\bar{w}_{[j]} + |\mathbf{p}^{\top} \mathbf{T}|\epsilon^w_{[j]} \}.
\end{align}
We now exploit \eqref{eq:box_SF_hull} to encode \eqref{eq:con_1_SF}-\eqref{eq:con_2_SF} as linear inequalities. 
Firstly, inequality~\eqref{eq:con_1_SF} holds by Equation~\eqref{eq:box_SF_hull} if and only if 
\begin{align}
&\hspace{-15pt}\sum_{t=0}^{s-1} \left( \max_{j \in \mathbb{I}_1^N}
\left\{\bar{G}_{[t]}B \bar{w}_{[j]} + |\bar{G}_{[t]}B| \epsilon^w_{[j]}\right\}
\right) + \label{eq:full_SF_for_Y}\nonumber\\
&\hspace{50pt}\left( 
\max_{j \in \mathbb{I}_1^N}
\left\{GD \bar{w}_{[j]} + |GD| \epsilon^w_{[j]}\right\}
\right) \nonumber \\
&\hspace{120pt}\leq g - \lambda\sum_{t=0}^{s-1} |\bar{G}_{[t]}|\1. 
\end{align}
To encode~\eqref{eq:full_SF_for_Y}, we introduce $\mathbf{Q}:=\{\mathbf{Q}_{[t]} \in \R^{m_{\mathcal{Y}}}, \ t \in \mathbb{I}_0^{s-1}\}$ and $\mathbf{r} \in \R^{m_{\mathcal{Y}}}$, along with the inequalities
\begin{align}
\label{eq:Qr_conditions}
\begin{matrix*}[l]\forall \ j \in \mathbb{I}_1^N, \vspace{3pt} \\ \forall \ t \in \mathbb{I}_0^{s-1} \end{matrix*}
\begin{cases}
\bar{G}_{[t]} B \bar{w}_{[j]}+|\bar{G}_{[t]}B|\epsilon^w_{[j]} \leq \mathbf{Q}_{[t]}, \vspace{5pt} \\
GD\bar{w}_{[j]}+|GD|\epsilon^w_{[j]} \leq \mathbf{r}.
\end{cases}
\end{align}
%
%
Then, the inequality in~\eqref{eq:full_SF_for_Y} holds if and only if there exists some $(\mathbf{Q},\mathbf{r})$ satisfying the inequalities in~\eqref{eq:Qr_conditions} along with
\begin{align}
\label{eq:Qr_conditions_2}
\sum_{t=0}^{s-1} \mathbf{Q}_{[t]i} + \mathbf{r}_{i} \leq g_i - \lambda\sum_{t=0}^{s-1}|\bar{G}_{[t]i}| \1, && \forall \ i \in \mathbb{I}_1^{m_{\mathcal{Y}}}.
\end{align}
Thus, we encode the support function inequality in~\eqref{eq:con_1_SF} for the disturbance set parameterization in~\eqref{eq:convex_hull_of_boxes} as the linear inequalities~\eqref{eq:Qr_conditions} and~\eqref{eq:Qr_conditions_2}.
%
%
%
%
%
%
Similarly, support function inequality~\eqref{eq:con_2_SF} holds according to Equation~\eqref{eq:box_SF_hull}
if and only if
\begin{align}
\label{eq:inclusion_in_BW}
\tilde{\I}_{n_x} B \bar{w}_{[j]}+|\tilde{\I}_{n_x} B|\epsilon^w_{[j]} \leq \gamma \1, && \forall \ j \in \mathbb{I}_1^N.
\end{align}
\subsection{Inclusions \eqref{eq:final_formulation_to_solve:con3} and \eqref{eq:final_formulation_to_solve:con4}}	
We encode inclusion~\eqref{eq:final_formulation_to_solve:con3}, i.e., $\0 \in \mathcal{W}$, by enforcing the inclusion $\0 \in \mathbb{W}(\bar{w}_{[1]},\epsilon^w_{[1]})$ for simplicity through
\begin{align}
\label{eq:0_in_hull}
\begin{bmatrix} \I_{n_w} \\ -\I_{n_w} \end{bmatrix} \0 \leq \begin{bmatrix} \epsilon^w_{[1]} \\\epsilon^w_{[1]} \end{bmatrix} + \begin{bmatrix} \I_{n_w} \\ -\I_{n_w} \end{bmatrix} \bar{w}_{[1]} 
\end{align}
%
%
since $\0 \in \mathbb{W}(\bar{w}_{[1]},\epsilon^w_{[1]})$ implies $\0  \in \mathcal{W}$ from \eqref{eq:convex_hull_of_boxes}.

\vspace{5pt}
In order to encode inclusion~\eqref{eq:final_formulation_to_solve:con4}, i.e., $\mathcal{Y} \subseteq \mathcal{S}(l,\mathcal{W}) \oplus \mathbb{B}(\epsilon),$ we recall from~\eqref{eq:l_step_reachable} that
$
\mathcal{S}(l,\mathcal{W})=\bigoplus_{t=0}^{l-1} CA^t B \mathcal{W}\oplus D \mathcal{W}
$
is the $l$-step reachable set.
Note that the inclusion holds according to~\eqref{eq:SF_prop_3}-\eqref{eq:set_inclusion_generic_compact} if and only if the support function inequality
\begin{align}
\label{eq:impossible_inclusion}
h_{\mathcal{Y}}(\mathbf{p})  \leq \sum_{t=0}^{l-1}h_{CA^t B\mathcal{W}}(\mathbf{p}) +h_{D\mathcal{W}}(\mathbf{p})+h_{\mathbb{B}(\epsilon)}(\mathbf{p})
\end{align}
is verified for all $\mathbf{p}\in \R^{n_y}$. If the hyperplane notation of $\mathcal{W}$ is known, then \cite[Theorem 1]{Sadraddini2019} can be used to to derive sufficient linear conditions for~\eqref{eq:impossible_inclusion}. Unfortunately, since the hyperplane notation is unknown a priori, we rely on Proposition~\ref{prop:basic_properties}($e$) to encode inclusion~\eqref{eq:final_formulation_to_solve:con4}. Denoting
$\{\mathrm{y}_{[i]}, \ i \in \mathbb{I}_1^{v_{\mathcal{Y}}}\} := \mathrm{vert}(\mathcal{Y}),$
and recalling the definition of $\mathcal{S}(l,\mathcal{W})$ from~\eqref{eq:l_step_reachable}, we know from Proposition~\ref{prop:basic_properties}($e$) that inclusion~\eqref{eq:final_formulation_to_solve:con4} holds if and only if
\begin{subequations}
\label{eq:inclusion_conditions_outer_Y}
\begin{align}
&\forall \ i \in \mathbb{I}_1^{v_{\mathcal{Y}}}, \ \scalemath{1.5}{\exists} \begin{cases}  
& \hspace{-5pt} \left\{\{\mathrm{w}^1_{[it]}, \ t \in \mathbb{I}_0^{l-1}\}, \ \mathrm{w}^2_{[i]}\right\} \in \mathcal{W},\\
& \hspace{-5pt} \ \mathrm{b}_{[i]} \in \mathbb{B}(\epsilon),
\end{cases}
\label{eq:inclusion_conditions_outer_Y:2} \\
& \text{such that }\mathrm{y}_{[i]}=\sum_{t=0}^{l-1} CA^{l-1-t}B\mathrm{w}_{[it]}^1+D\mathrm{w}_{[i]}^2+\mathrm{b}_{[i]}, \label{eq:inclusion_conditions_outer_Y:4}
\end{align}
\end{subequations}
%
where, for each $i \in \mathbb{I}_1^{v_{\mathcal{Y}}}$, feasible disturbance sequences in~\eqref{eq:inclusion_conditions_outer_Y:2} drive the output of System~\eqref{eq:system} in $l$-steps
to some
$$\bm{\mathrm{y}}_{[i]}(l):=\sum_{t=0}^{l-1} CA^{l-1-t}B\mathrm{w}_{[it]}^1+D\mathrm{w}_{[i]}^2,$$
that belongs in the vicinity of vertex $\mathrm{y}_{[i]}$ as $\mathrm{y}_{[i]} - \bm{\mathrm{y}}_{[i]}(l) \in \mathbb{B}(\epsilon).$
We will show in the sequel that the conditions in~\eqref{eq:inclusion_conditions_outer_Y} verifying inclusion~\eqref{eq:final_formulation_to_solve:con4} can be tractably encoded using necessary and sufficient conditions for the disturbance set parameterization in~\eqref{eq:convex_hull_of_boxes}. To that end, we make the following assumption.
%
%
%
%
%
%
\begin{assumption}
\label{ass:Y_vertices_known}
Vertices $\{\mathrm{y}_{[i]}, i \in \mathbb{I}_1^{v_{\mathcal{Y}}}\}$ of $\mathcal{Y}$ are given.
$\hfill\square$
\end{assumption}

In order to encode the conditions in~\eqref{eq:inclusion_conditions_outer_Y}, we need to enforce the point-wise constraints $\mathrm{w}^1_{[it]} \in \mathcal{W}$ and $\mathrm{w}^2_{[i]}\in \mathcal{W}$ in~\eqref{eq:inclusion_conditions_outer_Y:2}. According to the parameterization of $\mathcal{W}$ in~\eqref{eq:convex_hull_of_boxes}, these constraints can be enforced by guaranteeing that points $\mathrm{w}^1_{[it]}$ and $\mathrm{w}^2_{[i]}$ belong to the convex hull of boxes $\mathbb{W}(\bar{w}_{[j]},\epsilon^w_{[j]})$.
In the following result, we show that this condition is equivalent to enforcing $\mathrm{w}^1_{[it]}$ and $\mathrm{w}^2_{[i]}$ to belong to convex hulls of points, with each point belonging to a box $\mathbb{W}(\bar{w}_{[j]},\epsilon^w_{[j]})$.

\begin{proposition}
\label{prop:convhull_iff_prop}
Given the disturbance set parametrization in Equation~\eqref{eq:convex_hull_of_boxes}, there exists some $\mathfrak{w} \in \mathcal{W}$ if and only if there exist some $\mathfrak{w}_{[j]} \in \mathbb{W}(\bar{w}_{[j]},\epsilon^w_{[j]})$ for each $j \in \mathbb{I}_1^N$ such that $\mathfrak{w} \in \mathrm{ConvHull}(\mathfrak{w}_{[j]},j \in \mathbb{I}_1^N)$.  $\hfill\square$
\end{proposition}
\begin{proof}
Sufficiency follows by observing that $$\mathfrak{w} \in \mathrm{ConvHull}(\mathfrak{w}_{[j]} \in \mathbb{W}(\bar{w}_{[j]},\epsilon^w_{[j]}) ,j \in \mathbb{I}_1^N) \Rightarrow \mathfrak{w} \in \mathcal{W}$$ from Equation~\eqref{eq:convex_hull_of_boxes}.
For the necessary condition, we define $\{\mathfrak{v}_{[ji]},i \in \mathbb{I}_1^{2^{n_w}}\}:=\mathrm{vert}(\mathbb{W}(\bar{w}_{[j]},\epsilon^w_{[j]})),$
and note that if the point $\mathfrak{w} \in \mathcal{W}$, then there exist $\mathfrak{p}_{[ji]} \geq 0$ satisfying
\begin{align}
\label{eq:condition_sufficiency}
\mathfrak{w}=\sum_{j=1}^N \left(\sum_{i=1}^{2^{n_w}} \mathfrak{p}_{[ji]}\mathfrak{v}_{[ji]}\right), \qquad \sum_{j=1}^N \sum_{i=1}^{2^{n_w}} \mathfrak{p}_{[ji]}=1,
\end{align}
by convexity of $\mathcal{W}$. 
Then, we define $\hat{\mathfrak{p}}_{[j]}:=\sum_{i=1}^{2^{n_w}} \mathfrak{p}_{[ji]},$
and consider the two following cases:
\begin{align*}
&1) \text{ If } \hat{\mathfrak{p}}_{[j]}>0: \text{ Set } \hat{\mathfrak{v}}_{[j]}:=\left(\sum_{i=1}^{2^{n_w}} \mathfrak{p}_{[ji]}\mathfrak{v}_{[ji]}\right)/\hat{\mathfrak{p}}_{[j]}; \\
&2) \text{ If } \hat{\mathfrak{p}}_{[j]}=0: \text{ Select any } \hat{\mathfrak{v}}_{[j]} \in \mathbb{W}(\bar{w}_{[j]},\epsilon^w_{[j]}).
\end{align*}
We note that $\hat{\mathfrak{v}}_{[j]} \in \mathbb{W}(\bar{w}_{[j]},\epsilon^w_{[j]})$ in Case $1$, and Case $2$ $\Leftrightarrow \mathfrak{p}_{[ji]}=0, \ \forall \ i \in \mathbb{I}_1^{2^{n_w}}$. Finally, observe that~\eqref{eq:condition_sufficiency} can be rearranged as $\mathfrak{w}=\sum_{j=1}^N \hat{\mathfrak{p}}_{[j]} \hat{\mathfrak{v}}_{[j]}$ and $\sum_{j=1}^N \hat{\mathfrak{p}}_{[j]}=1$.
The proof concludes by setting $\mathfrak{w}_{[j]}=\hat{\mathfrak{v}}_{[j]} \in \mathbb{W}(\bar{w}_{[j]},\epsilon^w_{[j]})$.
\end{proof}
Using this result, we replace inclusions $\mathrm{w}^1_{[it]} ,\mathrm{w}^2_{[i]}\in \mathcal{W}$ in~\eqref{eq:inclusion_conditions_outer_Y:2} with the equivalent inclusions
%
\begin{align}
\label{eq:inner_polytope_approx}
\hspace{-10pt}
\mathrm{w}^1_{[it]} \in \mathcal{W}^1_{[it]}&:=\mathrm{ConvHull}(\bar{\mathrm{w}}^1_{[itj]} \in \mathbb{W}(\bar{w}_{[j]},\epsilon^w_{[j]}), j \in \mathbb{I}_1^N), \nonumber \vspace{2pt}\\
\ \mathrm{w}^2_{[i]} \in \mathcal{W}^2_{[i]}&:=\mathrm{ConvHull}(\bar{\mathrm{w}}^2_{[ij]} \in \mathbb{W}(\bar{w}_{[j]},\epsilon^w_{[j]}), j \in \mathbb{I}_1^N), 
\end{align}
by introducing variables $\bar{\mathrm{w}}^1_{[itj]},\bar{\mathrm{w}}^2_{[ij]}\in \mathbb{W}(\bar{w}_{[j]},\epsilon^w_{[j]})$.
Then, we write the conditions in~\eqref{eq:inclusion_conditions_outer_Y} equivalently as
\begin{subequations} %
\label{eq:inclusion_outside_Y}
\begin{align}
\hspace{-17pt}
\forall \ j \in \mathbb{I}_1^N, \	\forall \ t \in \mathbb{I}_0^{l-1}, \ 	\forall \ i \in \mathbb{I}_1^{v_{\mathcal{Y}}},\label{eq:inclusion_outside_Y:1} \\
\mathrm{y}_{[i]}=\sum_{t=0}^{l-1} CA^{l-1-t}B\mathrm{w}_{[it]}^1+D\mathrm{w}_{[i]}^2+\mathrm{b}_{[i]}, \label{eq:inclusion_outside_Y:2} \\
\mathrm{w}^1_{[it]}=\sum_{j=1}^N \beta^1_{[itj]}\bar{\mathrm{w}}^1_{[itj]}, \ 	\mathrm{w}^2_{[i]}=\sum_{j=1}^N \beta^2_{[ij]}\bar{\mathrm{w}}^2_{[ij]}, \label{eq:inclusion_outside_Y:3}\\ 	\bar{\mathrm{w}}^1_{[itj]} \in \mathbb{W}(\bar{w}_{[j]},\epsilon^w_{[j]}),  \ \bar{\mathrm{w}}^2_{[ij]} \in \mathbb{W}(\bar{w}_{[j]},\epsilon^w_{[j]}),
\label{eq:inclusion_outside_Y:4}\\
\sum_{j=1}^N \beta^1_{[itj]}=1, \beta^1_{[itj]} \geq 0, \ 	\sum_{j=1}^N \beta^2_{[ij]}=1, \beta^2_{[ij]} \geq 0, \label{eq:inclusion_outside_Y:5}\\
\mathrm{b}_{[i]} \in \mathbb{B}(\epsilon), \label{eq:inclusion_outside_Y:6}
\end{align}
\end{subequations}
in which the variables $\beta^1_{[itj]}$ and $\beta^2_{[ij]}$ are introduced to encode the convex-hull inclusions in \eqref{eq:inner_polytope_approx} through~\eqref{eq:inclusion_outside_Y:3} and \eqref{eq:inclusion_outside_Y:5}.	
\begin{table}
\centering
\resizebox{0.88\columnwidth}{!}{
\begin{tabular}{|c|c|c|c|}
\hline
Variable  & Dimension & Variable  & Dimension \\ \hline
$\bm{\mathrm{x}}$&  $2Nn_w+(s+1)m_{\mathcal{Y}}$ & $\bm{\beta}$&   $v_{\mathcal{Y}} \times N (l+1)$ \\ \hline
$\bm{\mathrm{w}}$ &  $v_{\mathcal{Y}} \times (l+1) n_w$ & $\bm{\mathrm{z}}$& $n_B+ v_{\mathcal{Y}} \times n_y$ \\ \hline
$\bar{\bm{\mathrm{w}}}$ & $v_{\mathcal{Y}} \times N(l+1) n_w$ & & \\ \hline 
\end{tabular}
}
\caption{Dimensions of variables defined in~\eqref{eq:variable_definitions}.}
\label{table:var_dim}
\end{table}
%
\begin{table}
\centering
\resizebox{0.92\columnwidth}{!}{
\begin{tabular}{|c|c|}
\hline
Constraint  & $\#$  \\ \hline
\eqref{eq:Qr_conditions}-\eqref{eq:inclusion_in_BW}, \eqref{eq:0_in_hull}&  $(s+1)m_{\mathcal{Y}}+2(n_x+n_w)$ Lin. ineq.  \\ \hline
\eqref{eq:inclusion_outside_Y:2} &  $v_{\mathcal{Y}} \times n_y$ Lin. eq. \\ \hline
\eqref{eq:inclusion_outside_Y:3} & $v_{\mathcal{Y}} \times (l+1) n_w$ Bilin. eq.\\ \hline
\eqref{eq:inclusion_outside_Y:4}&   $v_{\mathcal{Y}} \times 2 N (l+1) n_w$ Lin. ineq. \\ \hline
\eqref{eq:inclusion_outside_Y:5}& $v_{\mathcal{Y}} \times N (l+1)$ Lin. ineq., \ $v_{\mathcal{Y}} \times (l+1)$ Lin. eq \\ \hline
\eqref{eq:inclusion_outside_Y:6}& $v_{\mathcal{Y}} \times n_B$ Linear ineq.\\ \hline
\end{tabular}
}
\caption{Number of constraints.}
\label{table:con_dim}
\end{table}

Thus, we encode constraints~\eqref{eq:final_formulation_to_solve:con1}-\eqref{eq:final_formulation_to_solve:con4} as \eqref{eq:Qr_conditions}-\eqref{eq:inclusion_in_BW}, \eqref{eq:0_in_hull} and  \eqref{eq:inclusion_outside_Y} respectively. For simplicity of notation, we define
\begin{subequations}
\label{eq:variable_definitions}
\begin{align}
\bm{\mathrm{x}}&:=\{\{\bar{w}_{[j]},\epsilon^w_{[j]},j \in \mathbb{I}_1^N\}, \mathbf{Q}, \mathbf{r}\}, \label{eq:x_defn}\\
{\bm{\mathrm{w}}}&:=\{{\mathrm{w}}^1_{[it]},{\mathrm{w}}^2_{[i]}, \ i \in \mathbb{I}_1^{v_{\mathcal{Y}}}, t \in \mathbb{I}_0^{l-1}\}, \label{eq:w_defn} \\
\bar{\bm{\mathrm{w}}}&:=\{\bar{\mathrm{w}}^1_{[itj]},\bar{\mathrm{w}}^2_{[ij]}, \ i \in \mathbb{I}_1^{v_{\mathcal{Y}}}, t \in \mathbb{I}_0^{l-1},j \in \mathbb{I}_1^N\}, \label{eq:bar_w_defn} \\
\bm{\beta}&:=\{\beta^1_{[itj]},\beta^2_{[ij]}, \ i \in \mathbb{I}_1^{v_{\mathcal{Y}}}, t \in \mathbb{I}_0^{l-1},j \in \mathbb{I}_1^N\}, \label{eq:beta_defn}\\
\bm{\mathrm{z}}&:=[\epsilon^{\top} \ \mathrm{b}_{[1]}^{\top} \ \cdots \ \mathrm{b}_{v_{\mathcal{Y}}}^{\top}]^{\top}, \label{eq:z_defn}
\end{align}
\end{subequations}
and denote $\bm{\mathrm{v}}:=\{\bm{\mathrm{x}},{\bm{\mathrm{w}}},\bar{\bm{\mathrm{w}}},\bm{\beta},	\bm{\mathrm{z}}\}$. 
Over these variables, we denote the constraints in \eqref{eq:Qr_conditions}-\eqref{eq:inclusion_in_BW}, \eqref{eq:0_in_hull} and  \eqref{eq:inclusion_outside_Y} as
\begin{subequations}
\label{eq:constraint_definitions}
\begin{align}
&\bm{\mathrm{x}} \text{ satisfies~\eqref{eq:Qr_conditions}-\eqref{eq:inclusion_in_BW}, \eqref{eq:0_in_hull} } &\Leftrightarrow& \quad  \bm{\mathrm{A}} \bm{\mathrm{x}} \leq \bm{\mathrm{b}}, \\
&\bm{\mathrm{w}},\bm{\mathrm{{z}}} \text{ satisfy~\eqref{eq:inclusion_outside_Y:2}} &\Leftrightarrow& \quad \bm{\mathrm{C}}_{\bm{\mathrm{w}}} {\bm{\mathrm{w}}} +\bm{\mathrm{C}}_{{\bm{\mathrm{z}}}}{{\bm{\mathrm{z}}}} = \bm{\mathrm{h}}, \\
&\bm{\mathrm{w}},\bar{\bm{\mathrm{{w}}}},\bm{\beta} \text{ satisfy~\eqref{eq:inclusion_outside_Y:3}} &\Leftrightarrow& \quad \bm{\mathrm{g}}(\bar{\bm{\mathrm{{w}}}},\bm{\beta})=\bm{\mathrm{w}}, \\
&\bm{\mathrm{x}},\bar{\bm{\mathrm{{w}}}} \text{ satisfy~\eqref{eq:inclusion_outside_Y:4}} &\Leftrightarrow& \quad \bm{\mathrm{D}}_{\bm{\mathrm{x}}} {\bm{\mathrm{x}}} +\bm{\mathrm{D}}_{\bar{\bm{\mathrm{w}}}}{\bar{\bm{\mathrm{w}}}}\leq \0, \\
&\bm{\beta} \text{ satisfies~\eqref{eq:inclusion_outside_Y:5}} &\Leftrightarrow& \quad \bm{\beta}\geq \0, \ \ \bm{\mathrm{T}}_{\bm{\beta}} \bm{\beta}=\1, \\
&\bm{\mathrm{z}} \text{ satisfies~\eqref{eq:inclusion_outside_Y:6}} &\Leftrightarrow& \quad \bm{\mathrm{E}}_{{\bm{\mathrm{z}}}}{{\bm{\mathrm{z}}}} \leq \0.
\end{align}
\end{subequations}
Finally, we define the cost vector $\bm{\mathrm{c}}:=[\1^{\top}_{n_{B}} \quad \0^{\top}_{n_y v_{\mathcal{Y}}}]^{\top}$, such that $\bm{\mathrm{c}}^{\top} \bm{\mathrm{z}}=\norm{\epsilon}_1.$
Then, we write Problem~\eqref{eq:final_formulation_to_solve} as
\begin{subequations}
\label{eq:equivalent_SF_form_nonlinear}
\begin{align}
\hspace{-5pt}
&\min_{\bm{\mathrm{v}}=\{\bm{\mathrm{x}},\bm{\mathrm{w}},\bar{\bm{\mathrm{w}}},\bm{\beta},\bm{\mathrm{z}}\}} \quad \bm{\mathrm{c}}^{\top} \bm{\mathrm{z}} \\ & \qquad \ \ \text{s.t.}  \qquad \
\bm{\mathrm{A}} \bm{\mathrm{x}} \leq \bm{\mathrm{b}}, \label{eq:equivalent_SF_form_nonlinear:1} \\
&\qquad \qquad \ \qquad \bm{\mathrm{D}}_{\bm{\mathrm{x}}} {\bm{\mathrm{x}}} +\bm{\mathrm{D}}_{\bar{\bm{\mathrm{w}}}}{\bar{\bm{\mathrm{w}}}} \leq \0,  \label{eq:equivalent_SF_form_nonlinear:2}\\
&\qquad \qquad \ \qquad	\bm{\mathrm{C}}_{\bm{\mathrm{w}}} {\bm{\mathrm{w}}} +\bm{\mathrm{C}}_{{\bm{\mathrm{z}}}}{{\bm{\mathrm{z}}}} = \bm{\mathrm{h}},  \label{eq:equivalent_SF_form_nonlinear:3}\\
&\qquad \qquad \ \qquad \bm{\mathrm{E}}_{{\bm{\mathrm{z}}}}{{\bm{\mathrm{z}}}} \leq \0, \label{eq:equivalent_SF_form_nonlinear:4}\\
&\qquad \qquad \ \qquad	\bm{\beta}\geq \0, \ \ \bm{\mathrm{T}}_{\bm{\beta}} \bm{\beta}=\1, \label{eq:equivalent_SF_form_nonlinear:5}\\
&\qquad \qquad \  \qquad \bm{\mathrm{g}}(\bar{\bm{\mathrm{{w}}}},\bm{\beta})=\bm{\mathrm{w}}. \label{eq:equivalent_SF_form_nonlinear:6}
\end{align}
\end{subequations}
The number of variables and constraints defining Problem~\eqref{eq:equivalent_SF_form_nonlinear} are shown in Tables~\ref{table:var_dim} and \ref{table:con_dim} respectively, in which we observe that the number of variables and constraints scale linearly with the number of vertices $v_{\mathcal{Y}}$ of the output constraint set $\mathcal{Y}$. This problem is composed of a linear objective and polyhedral constraints, along with bilinear equality constraints in \eqref{eq:equivalent_SF_form_nonlinear:6} resulting from~\eqref{eq:inclusion_outside_Y:3}. 
Since this problem is smooth, it can be solved to local optimality using any off-the-shelf nonlinear programing (NLP) solver~\cite{NoceWrig06}. 

{\color{black}
\section{Approximate solutions of Problem~\eqref{eq:final_formulation_to_solve}}
\label{sec:approx_soln_methods}
While an NLP approach can be used to solve Problem~\eqref{eq:equivalent_SF_form_nonlinear}, the implementation of NLP solvers can be cumbersome in practice. Moreover, the quality of solutions computed by an NLP solver on a non-convex problem depends on the initial point. Hence, we now present a simple LP-based algorithm to approximately solve Problem~\eqref{eq:final_formulation_to_solve}. The output of this algorithm can be used to initialize an NLP solver to solve Problem~\eqref{eq:equivalent_SF_form_nonlinear}.}

The algorithm is based on the observation that the bilinear equality~\eqref{eq:equivalent_SF_form_nonlinear:6} can be reduced to a linear equality by fixing the value of $\bm{\beta}$. This reduction of the bilinear equality leads to a simplification of Problem~\eqref{eq:equivalent_SF_form_nonlinear} to the LP
%
\begin{align*}
\hspace{-0pt}
\mathbb{P}(\bm{\beta}_{\mathrm{f}})
\begin{cases}
\left\{\bm{\mathrm{x}}_*,\bm{\mathrm{w}}(\bm{\beta}_{\mathrm{f}}),\bar{\bm{\mathrm{w}}}_*,\bm{\mathrm{z}}(\bm{\beta}_{\mathrm{f}})\right\}&\hspace{-5pt}:= \underset{{\bm{\mathrm{x}},\bm{\mathrm{w}},\bar{\bm{\mathrm{w}}},\bm{\mathrm{z}}}}{\arg\min} \quad \bm{\mathrm{c}}^{\top} \bm{\mathrm{z}} \\ & \hspace{-5pt}\qquad \text{s.t.}  \quad
\eqref{eq:equivalent_SF_form_nonlinear:1}-\eqref{eq:equivalent_SF_form_nonlinear:4}, \\ &\qquad \qquad \ \hspace{-5pt} \bm{\mathrm{g}}(\bar{\bm{\mathrm{{w}}}},\bm{\beta}_{\mathrm{f}})=\bm{\mathrm{w}},
\end{cases}
\end{align*}
where $\bm{\beta}_{\mathrm{f}}$ is some value of $\bm{\beta}$ that satisfies constraint~\eqref{eq:equivalent_SF_form_nonlinear:5}

A special case of arises when the number of boxes parameterizing the disturbance set is equal to the number of vertices of the output constraint set $\mathcal{Y}$, i.e., $N=v_{\mathcal{Y}}$. Then, problem $\mathbb{P}(\bm{\beta}_{\mathrm{f}})$ can be solved with the components of $\bm{\beta}_{\mathrm{f}}$ selected as
\begin{align}
\hspace{-5pt}
\label{eq:special_case_heuristic}
\forall \ i \in \mathbb{I}_1^{v_{\mathcal{Y}}}, \forall \  t \in \mathbb{I}_0^{s-1}, \	\beta^1_{[itj]},\ \beta^2_{[ij]} = \begin{cases} 1, \text{  if  } i=j, \\
0, \text{  otherwise}.
\end{cases}
\end{align}
This is equivalent to enforcing the disturbance sequences $\{\{\mathrm{w}^1_{[it]}, \ t \in \mathbb{I}_0^{l-1}\}, \ \mathrm{w}^2_{[i]}\}$ corresponding to vertex $\mathrm{y}_{[i]}$ of the disturbance set $\mathcal{Y}$ inside the box $\mathbb{W}(\bar{w}_{[i]},\epsilon^w_{[i]})$, instead of in the set $\mathcal{W}$ as done in~\eqref{eq:inclusion_conditions_outer_Y:2}. 

While $\mathbb{P}(\bm{\beta}_{\mathrm{f}})$ provides an efficient way to approximate Problem~\eqref{eq:equivalent_SF_form_nonlinear}, conservativeness can be reduced further by also optimizing over $\bm{\beta}$. To this end, we observe that Problem~\eqref{eq:equivalent_SF_form_nonlinear} reduces to an LP also for a fixed value of $\bar{\bm{\mathrm{{w}}}}$. Based on this observation, we propose to solve the LP
\begin{align*}
\hspace{-0pt}
\mathbb{Q}(\bar{\bm{\mathrm{w}}}_*)
\begin{cases}
\left\{\bm{\mathrm{w}}_*,\bm{\mathrm{z}}_*,\bm{\beta}_*\right\}:= & \ \underset{{\bm{\mathrm{w}},\bm{\mathrm{z}}},\bm{\beta}}{\arg\min} \quad \bm{\mathrm{c}}^{\top} \bm{\mathrm{z}} \\ & \qquad \text{s.t.}  \quad
\eqref{eq:equivalent_SF_form_nonlinear:3}-\eqref{eq:equivalent_SF_form_nonlinear:5}, \\ &\qquad \qquad \  \bm{\mathrm{g}}(\bar{\bm{\mathrm{{w}}}}_*,\bm{\beta})=\bm{\mathrm{w}}, 
\end{cases}
\end{align*}
where $\bar{\bm{\mathrm{w}}}_*$ is an optimizer of problem $\mathbb{P}(\bm{\beta}_{\mathrm{f}})$. Using LPs $\mathbb{P}(\bm{\beta}_{\mathrm{f}})$ and $\mathbb{Q}(\bar{\bm{\mathrm{w}}}_*)$, we define an alternating-minimization procedure in Algorithm~\ref{alg:alg2} to approximately solve Problem~\eqref{eq:equivalent_SF_form_nonlinear}, in which we select $\bm{\beta}_{\mathrm{f}}=\bm{\beta}_*$ and repeat the steps. We use the superscript $^{[\iota]}$ to denote the iteration index.
This procedure can be interpreted as follows. Using $\mathbb{P}(\bm{\beta}_{\mathrm{f}})$, a disturbance set $\mathcal{W}$ characterized by variables $\bm{\mathrm{x}}$ is computed by selecting the disturbance sequences $\bm{\mathrm{w}}(\bm{\beta}_{\mathrm{f}})$. Then using $\mathbb{Q}(\bar{\bm{\mathrm{w}}}_*)$, these sequences are updated to $\bm{\mathrm{w}}_*$ by optimizing over $\bm{\beta}$, while the sets $\mathcal{W}^1_{[it]}$ and $\mathcal{W}^2_{[i]}$ characterized by $\bar{\bm{\mathrm{w}}}_*$ are kept fixed. 
\begin{algorithm}[t]
\caption{Alternating-minimization for Problem \eqref{eq:equivalent_SF_form_nonlinear} }
\begin{algorithmic}[1]
\State \textbf{Input:} Initial $\bm{\beta}=\bm{\beta}^{[0]}_*$ satisfying \eqref{eq:equivalent_SF_form_nonlinear:5};
\State \textbf{Set} $\zeta >0$, $conv=0$, $\iota=1$;
\While{$conv=0$}
\State Solve $\mathbb{P}(\bm{\beta}_*^{[\iota-1]})$ for $\left\{\bm{\mathrm{x}}^{[\iota]}_*,\bm{\mathrm{w}}(\bm{\beta}_*^{[\iota-1]}),\bar{\bm{\mathrm{w}}}^{[\iota]}_*,\bm{\mathrm{z}}(\bm{\beta}_*^{[\iota-1]})\right\}$; 
\State Solve $\mathbb{Q}(\bar{\bm{\mathrm{w}}}^{[\iota]}_*)$ for $\left\{\bm{\mathrm{w}}^{[\iota]}_*,\bm{\mathrm{z}}^{[\iota]}_*,\bm{\beta}^{[\iota]}_*\right\}$;  
\If{$\iota > 1$ and $\bm{\mathrm{c}}^{\top} \bm{\mathrm{z}}_*^{[\iota]} \geq \bm{\mathrm{c}}^{\top}\bm{\mathrm{z}}_*^{[\iota-1]} - \zeta$,} 
\State $conv\leftarrow 1$;
\Else
\State $\bm{\mathrm{v}}^{[\iota]}_*\leftarrow \left\{\bm{\mathrm{x}}^{[\iota]}_*,\bm{\mathrm{w}}_*^{[\iota]},\bar{\bm{\mathrm{w}}}^{[\iota]}_*,\bm{\beta}^{[\iota]}_*,\bm{\mathrm{z}}_*^{[\iota]}\right\}$, $\iota \leftarrow \iota+1$
;
\EndIf
\EndWhile
\State \textbf{Output:} $\bm{\mathrm{v}}^{[\iota]}_*$
\end{algorithmic}
\label{alg:alg2}
\end{algorithm}
\begin{proposition}
Algorithm~\ref{alg:alg2} terminates in finite time for any $\zeta>0$ with feasible iterates. $\hfill\square$
\end{proposition}
\begin{proof}
For any $\iota>0$, 
since $(\bm{\mathrm{w}}(\bm{\beta}^{[\iota-1]}_*),\bm{\mathrm{z}}(\bm{\beta}^{[\iota-1]}_*))$ computed by $\mathbb{P}(\bm{\beta}^{[\iota-1]}_*)$ are feasible for $\mathbb{Q}(\bar{\bm{\mathrm{w}}}^{[\iota]}_*)$, whose optimizers $(\bm{\mathrm{w}}^{[\iota]}_*,\bm{\mathrm{z}}^{[\iota]}_*)$ are in turn feasible for $\mathbb{P}(\bm{\beta}_*^{[\iota]})$, the inequalities
\begin{align}
\label{eq:alg_inequalities}
\bm{\mathrm{c}}^{\top} \bm{\mathrm{z}}(\bm{\beta}^{[\iota-1]}_*) \geq \bm{\mathrm{c}}^{\top} \bm{\mathrm{z}}_*^{[\iota]} \geq \bm{\mathrm{c}}^{\top} \bm{\mathrm{z}}(\bm{\beta}^{[\iota]}_*) \geq \bm{\mathrm{c}}^{\top} \bm{\mathrm{z}}_*^{[\iota+1]}
\end{align}
hold, such that $\bm{\mathrm{c}}^{\top} \bm{\mathrm{z}}(\bm{\beta}^{[\iota]}_*)$ and $\bm{\mathrm{c}}^{\top} \bm{\mathrm{z}}_*^{[\iota]}$ are nonincreasing in $[\iota]$.  By construction of Problem~\eqref{eq:equivalent_SF_form_nonlinear}, we know that $\bm{\mathrm{c}}^{\top} \bm{\mathrm{z}} \geq 0$ for all feasible $\bm{\mathrm{z}}$. Thus,  $\bm{\mathrm{c}}^{\top} \bm{\mathrm{z}}(\bm{\beta}^{[\iota]}_*)$ and $\bm{\mathrm{c}}^{\top} \bm{\mathrm{z}}_*^{[\iota]}$ are bounded below, such that for every $\zeta>0$, there exists some $\iota<\infty$ such that $\bm{\mathrm{c}}^{\top} \bm{\mathrm{z}}_*^{[\iota]} \geq \bm{\mathrm{c}}^{\top}\bm{\mathrm{z}}_*^{[\iota-1]} - \zeta$
holds, concluding the proof of finite termination. Feasibility of iterates $\bm{\mathrm{v}}^{[\iota]}_*$ follows since $\mathbb{P}(\bm{\beta}_*^{[\iota-1]})$ and $\mathbb{Q}(\bar{\bm{w}}^{[\iota]}_*)$ enforce the same constraints 
as Problem~\eqref{eq:equivalent_SF_form_nonlinear}.
\end{proof}
{\color{black}
\begin{remark}
The use of set parameterizations that have a closed-form linear expression of support functions is a viable alternative to the disturbance set parameterization in~\eqref{eq:convex_hull_of_boxes}. Affine transformations of $p$-norm balls, such as zonotopes, are a specific class of set parameterizations that admit such expressions. Additionally, by leveraging Proposition~\ref{prop:conv_hull}, we can represent the disturbance set as a convex hull of affinely transformed $p$-norm balls while still maintaining a linear inclusion encoding. We refer the interested reader to~\cite[Chapter 5]{MulagaletiThesis}, in which the disturbance set is parameterized as a convex hull of fixed-orientation zonotopes and exploits results from~\cite{Guibas2003} to formulate Problem~\eqref{eq:final_formulation_to_solve} as an LP.
$\hfill\square$
\end{remark}
}
\section{Numerical examples}
%
\label{sec:numerical_examples}
We now present numerical examples to illustrate the main ideas of our approach, compare performance with the state of the art, and demonstrate an application of the methods for control design. For further comparisons, we refer the reader to~\cite[Chapter 5]{MulagaletiThesis}.
Within these examples, we perform SOCP computations in Algorithm~\ref{alg:SAL_solving_procedure}  using MOSEK~\cite{mosek}, and LP computations in Algorithm~\ref{alg:alg2} using Gurobi~\cite{gurobi}.
We use the output of Algorithm~\ref{alg:alg2} to initialize the Primal-Dual Interior Point solver IPOPT \cite{Wchter2005} to solve Problem~\eqref{eq:equivalent_SF_form_nonlinear}, and use the MPT-toolbox~\cite{MPT3} for plotting the sets.
The computations were performed on a laptop with an Intel i7-7500U processor and $16$GB of RAM running MATLAB R2017b on Ubuntu 16.04. 
\subsection{Simple Illustrative Example}
\label{example:illustrative_example}
We consider the randomly generated system
\begin{align*}
&\scalemath{0.9}{
A=\begin{bmatrix}  -0.5844 &  -0.2378 &  -0.2015 \\
-0.2378  &  0.0368   & 0.6915 \\ 
-0.2015 &   0.6915&   -0.0162\end{bmatrix}, \ } 
\scalemath{0.9}{B = \begin{bmatrix}     0   &  0.8974 \\
0  & -1.8597 \\
0.8903  &   0.9479 \end{bmatrix}}, \\
&\scalemath{0.9}{
C=\begin{bmatrix}       0  &  2.0091 &  -0.1402 \\
-0.9894    &     0  &  1.1447 \end{bmatrix}, \ } 
\quad \scalemath{0.9}{
D=\begin{bmatrix}    -0.8078 &        0 \\
0.9676 &   0.6751 \end{bmatrix},} 	
\end{align*}
with output constraint set $\mathcal{Y}=\{y:Gy \leq \1\}$ defined with
\begin{align*}
\scalemath{0.95}{
G =\begin{bmatrix} -0.4489 &  -1.9691 &   1.0364 &   1.4018 &  -0.9868\\ 
2.1848  &  1.2596 &   0.8726  & -0.3397  & -2.0995 \end{bmatrix}^{\top}.
}
\end{align*}
%
%
%
%
%
%
We select $\mu=10^{-3}$ and $\gamma=0.2$ in Algorithm~\ref{alg:SAL_solving_procedure}. 
{\color{black} The choices of $\gamma$ and $\mu$ are motivated by the observations in Remarks~\ref{remark:gamma_condition}-\ref{remark:slow_system}, i.e., based on $\mathcal{Y}$, $C$ and $\rho(A)$.}
Using the MOSEK~\cite{mosek} SOCP solver, we converge in $0.19$s with $s=59$, $\alpha= 6.789 \times 10^{-4}$, $\lambda=6.784 \times 10^{-5}$. 
%
%
%
Then, we select $N=4$ boxes in $\R^2$ to parametrize $\mathcal{W}$ in \eqref{eq:convex_hull_of_boxes}, and use $l=s=59$ to formulate \eqref{eq:inclusion_outside_Y}. Finally, we choose a uniform $6$-sided polytope in $\R^2$ to define $\mathbb{B}(\epsilon)$, using which we formulate Problem~\eqref{eq:equivalent_SF_form_nonlinear}. {\color{black} This choice allows us to maximize coverage of $\mathcal{Y}$ in all directions.}

To address Problem~\eqref{eq:equivalent_SF_form_nonlinear}, we adopted Algorithm~\ref{alg:alg2} with an initialization of $\bm{\beta}$ as $\beta^1_{[itj]},\beta^2_{[ij]}=1/N$. The algorithm converged after $10$ iterations with $\norm{\epsilon}_1=1.0962$, with average iteration time of $0.3792$s. The resulting disturbance set and the corresponding output reachable set obtained at termination of  are plotted in Figure~\ref{figure:Illustrative} using thick blue lines, along with the progression of $\norm{\epsilon}_1$ over the iterations.
%
We initialize IPOPT with the Algorithm~\ref{alg:alg2}'s output to solve Problem~\eqref{eq:equivalent_SF_form_nonlinear} and obtain $\norm{\epsilon}_1=1.055$ at termination. The disturbance set and the corresponding output reachable set at termination are shown in Figure~\ref{figure:Illustrative}.  we
observe that the only contribution to the Lagrangian Hessian stems from the bilinear constraints, which are nonzero only on off-diagonal blocks. Because IPOPT computes a positive-definite Hessian approximation by simply adding a positive diagonal matrix, we observed that it is beneficial to neglect the contribution to the Lagrangian Hessian of these bilinear terms. We therefore pass a zero Hessian approximation to IPOPT. In this case, IPOPT converges in $281$ iterations in $14.7236$s. Passing the exact Lagrangian Hessian results in convergence in $994$ iterations, taking $69.1807$s and producing $\norm{\epsilon}_1=1.0826$.

We compare our approach to the explicit RPI (ERPI) set method proposed in~\cite{Mulagaleti2020} for solving Problem~\eqref{eq:orig_problem_to_solve}. ERPI uses fixed normal vectors to parameterize an RPI set and a disturbance set, resulting in $\norm{\epsilon}_1=1.2039$ when using $163$ and $16$ hyperplanes, respectively. Our implicit RPI set approach achieves reduced conservativeness, as shown in Figure~\ref{figure:Illustrative}. While the solution of the ERPI approach can be refined by increasing the number of hyperplanes parameterizing the sets, the resulting computational complexity would greatly increase due to the quadratic increase in variables and constraints.

\begin{figure}[t]
\centering
\hspace*{-0.4cm}
{
\resizebox{0.45\textwidth}{!}
{
\begin{tikzpicture}
\begin{scope}[xshift=0cm]
\node[draw=none,fill=none](tanks_fig) {\includegraphics[trim=3 0 0 0,clip,scale=2]{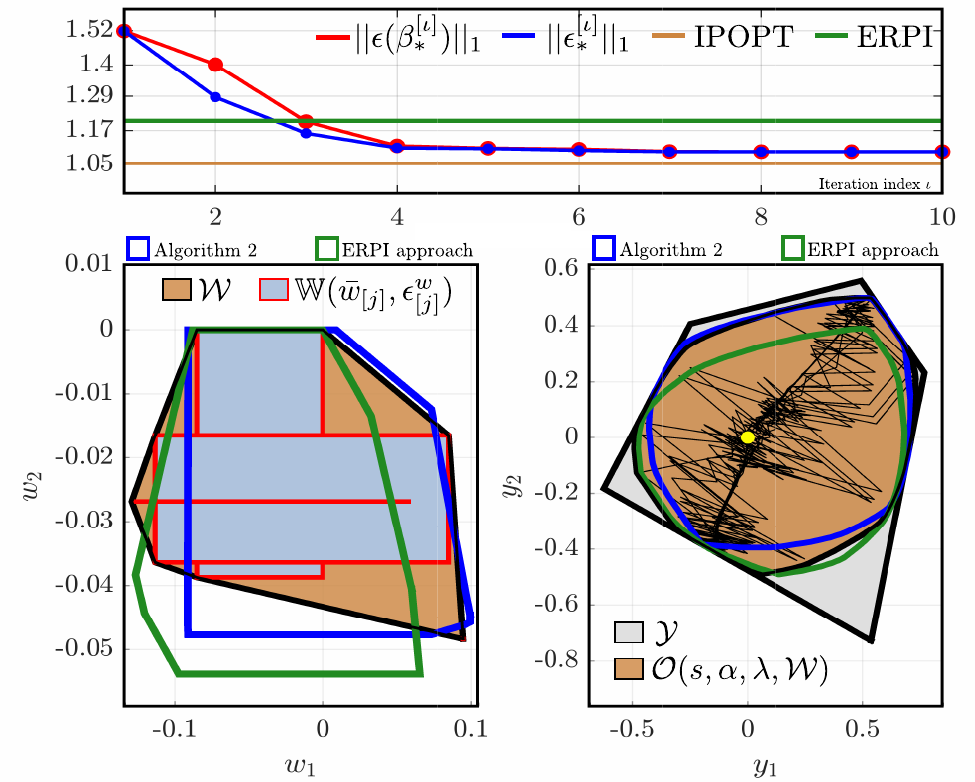}};
\end{scope}
\end{tikzpicture}}
}
\captionsetup{width=1\linewidth}
\caption{
(Top)  Convergence of Algorithm~\ref{alg:alg2} at $\iota=10$ with $\norm{\epsilon}_1=1.0962$;
(Bottom-Left) Disturbance set $\mathcal{W}$ with boxes $\WS[\bar{w}_{[j]},\epsilon^w_{[j]}]$ in \eqref{eq:convex_hull_of_boxes}. {\color{black} The bottom right vertex corresponds to box $j=2$ with $\epsilon^w_{[j]}=\0$ and $\bar{w}^{\top}_{[j]}=[0.0951 \ -0.0481]$}; (Bottom-Right) Output constraint set $\mathcal{Y}$ and the output reachable set $\mathcal{O}(s,\alpha,\lambda,\mathcal{W})$. Black lines are the output trajectories $y(t)$ of the system when initialized with $x(0)=\0$ (Yellow dot) and subject to random inputs $w \in \mathcal{W}$. We observe that $y \in \mathcal{O}(s,\alpha,\lambda,\mathcal{W}) \subset \mathcal{Y}$ holds. {\color{black} The thick blue lines indicate the solution at the termination of Algorithm~\ref{alg:alg2}, which is the initial guess to IPOPT for solving Problem~\eqref{eq:equivalent_SF_form_nonlinear}. The thick green lines indicate the sets computed using the ERPI approach in~\cite{Mulagaleti2020}, with the disturbance and RPI sets parameterized with $16$ and $163$ hyperplanes respectively.}}
\label{figure:Illustrative}
\end{figure}
%

We analyze the impact of the number of boxes $N$ parameterizing $\mathcal{W}$ in~\eqref{eq:convex_hull_of_boxes} when solving Problem~\eqref{eq:equivalent_SF_form_nonlinear}. We vary $N$ from $1$ to $10$, and for each $N$, we run Algorithm~\ref{alg:alg2} with $\beta^1_{[itj]},\beta^2_{[ij]}=1/N$, the output of which we use to initialize IPOPT. Figure~\ref{figure:effect_of_N} shows the resulting $\norm{\epsilon}_1$ values and CPU times. The trend of reducing $\norm{\epsilon}_1$ as $N$ increases is generally observed but not strictly monotonic due to the nonlinear optimization procedure's nature. Algorithm~\ref{alg:alg2} takes less than $10$s in the majority of cases, while IPOPT requires around $30$s for larger $N$ plus the initialization time using Algorithm~\ref{alg:alg2}. Using Algorithm~\ref{alg:alg2} output as IPOPT's initial guess can produce an acceptable solution for Problem~\eqref{eq:equivalent_SF_form_nonlinear}. The values of $\norm{\epsilon}_1$ are higher when IPOPT is initialized with an all-zero vector, demonstrating the importance of a good initial guess. {\color{black} For example, in the case of $N=4$ and all variables initialized to zero, IPOPT converges with $\norm{\epsilon}_1=1.0927$ in $328$ iterations instead of $\norm{\epsilon}_1=1.055$ in $281$ iterations.}

\begin{figure}[t]
\centering
{
\resizebox{0.48\textwidth}{!}
{
\begin{tikzpicture}
\begin{scope}[xshift=0cm]
\node[draw=none,fill=none](tanks_fig) {\includegraphics[trim=5 0 0 0,clip,scale=3.5]{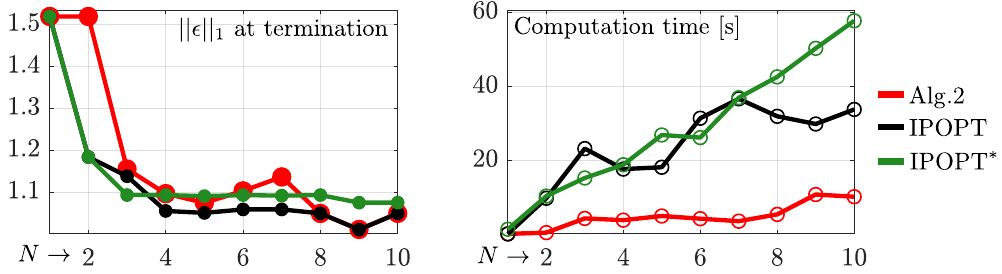}};
\end{scope}
\end{tikzpicture}}
}
\captionsetup{width=1\linewidth}
\caption{Effect of number of boxes $N$ parameterizing the disturbance set $\mathcal{W}$ on Problem~\eqref{eq:equivalent_SF_form_nonlinear}. The nonmonotonicity of $\norm{\epsilon}_1$ is because of the nonlinear nature of the problem. However, as expected in general smaller values of $\norm{\epsilon}_1$ are obtained as $N$ increases. Red lines indicate solution of IPOPT when initialized with output of Algorithm~\ref{alg:alg2}, and green lines indicate solution with initial guess all-zero. Observe that smaller values of $\norm{\epsilon}_1$ are computed if IPOPT is initialized with the solution of Algorithm~\ref{alg:alg2}.}
\label{figure:effect_of_N}
\end{figure}

\subsection{Comparison with the approach of~\cite{Mulagaleti2020}}
\label{sec:state_of_the_art}
We compare the performance of our implicit RPI (IRPI) approach with the ERPI approach proposed in~\cite{Mulagaleti2020} to solve Problem~\eqref{eq:orig_problem_to_solve} for $24$ randomly generated systems of dimensions listed in Table~\ref{table:comparison_examples}. Figure~\ref{figure:comparison_ERPI}-(Top) illustrates the values of $\rho(A)$ for each system. We set $\mathcal{Y}=\mathcal{B}_{\infty}^{n_y}$ and use the matrix $H=[\I \ -\I]^{\top}$ to define $\mathbb{B}(\epsilon)$ for all examples.
We employ Algorithm~\ref{alg:SAL_solving_procedure} to compute the RPI set parameters $(s,\alpha,\lambda)$ with $\mu=10^{-2}$ and $\gamma=1$. The parameter $s$ values are listed in Table~\ref{table:comparison_examples}, and $(\alpha,\lambda)$ values are plotted in Figure~\ref{figure:comparison_ERPI}-(Top). The average runtime of Algorithm~\ref{alg:SAL_solving_procedure} across all examples is $0.3411$s. We set $N=5$ to parameterize the disturbance set, and choose $l=s$ to parameterize $\mathcal{S}(l,\mathcal{W})$. We use Algorithm~\ref{alg:alg2} to compute an initial point for IPOPT to solve Problem~\eqref{eq:equivalent_SF_form_nonlinear}. We label the total runtime of Algorithm~\ref{alg:alg2} and IPOPT to solve problem~\eqref{eq:equivalent_SF_form_nonlinear} as $\mathrm{t}_{\mathrm{IRPI}}$, and the objective value as $\norm{\epsilon_{\mathrm{IRPI}}}_1$.
%

Regarding the ERPI approach of~\cite{Mulagaleti2020}, we recall that the RPI and disturbance sets are parameterized with fixed normal vectors as $
\mathcal{X}_{\mathrm{RPI}}(\mathcal{W}):=\{x: E_i x \leq \epsilon^x_i, \ i \in \mathbb{I}_1^{m_X}\}$ and $\mathcal{W}:=\{w: F_i x \leq \epsilon^w_i, \ i \in \mathbb{I}_1^{m_W}\}$ respectively.
%
The dimensions $(m_X,m_W)$ are shown in Table~\ref{table:comparison_examples}. We solve the optimization problem resulting from the ERPI approach using the 
specialized smoothening-based interior-point algorithm presented in~\cite{Mulagaleti2022_SmootheningPDIP}. We label the runtime of this algorithm as $\mathrm{t}_{\mathrm{ERPI}}$, and the objective value as $\norm{\epsilon_{\mathrm{ERPI}}}_1$.

Figure~\ref{figure:comparison_ERPI}-(Bottom) shows the ratios of objective values $\norm{\epsilon}_1$ and solution times between our IRPI approach and the ERPI approach. We observe that our IRPI approach yields much smaller $\norm{\epsilon}_1$ values and consumes significatly less time. The smaller $\norm{\epsilon}_1$ values stem from our approach not enforcing a specific RPI set representation a priori, while the reduced computation time is due to the cheap support function evaluations because of the parameterization in~\eqref{eq:convex_hull_of_boxes}. The minimum value of $\norm{\epsilon_{\mathrm{ERPI}}}_1/\norm{\epsilon_{\mathrm{IRPI}}}_1$ is $1.0042$ for Example 8.
%

At the output of the ERPI approach, we compute the approximation error $\mu_{\mathrm{ERPI}}$ of $\mathcal{X}_{\mathrm{RPI}}(\mathcal{W})$ with respect to the mRPI set $\mathcal{X}_{\mathrm{m}}(\mathcal{W})$. We plot these values in Figure~\ref{figure:comparison_ERPI}-(Top) in blue, in which we observe that the ERPI set approximation error is larger than $\mu=10^{-2}$. We perform the same operation to compute $\mu_{\mathrm{IRPI}}$ using the solution of the IRPI approach, and as expected we obtain  $\mu_{\mathrm{IRPI}} \leq 10^{-2}$ (Shown in green). While $\mu_{\mathrm{ERPI}}$ can be reduced with larger $m_X$, the computational complexity increases quadratically. Hence, the IRPI approach is a more attractive alternative to tackle Problem~\eqref{eq:orig_problem_to_solve}.

Finally, we use our IRPI approach to tackle Problem~\eqref{eq:orig_problem_to_solve} for higher dimensional systems. Because of the large values of $n_x$, the large number of hyperplanes $m_X$ required to represent an RPI set in the ERPI approach results in the system running out of memory, thus effectively failing to tackle Problem~\eqref{eq:orig_problem_to_solve}. On the other hand, the IRPI approach succeeds in tackling Problem~\eqref{eq:orig_problem_to_solve}
The system details and the results of the IRPI approach are shown in Table~\ref{table:higher_dimension_systems}, in which we select the parameters $(\mu,\gamma,\mathcal{Y},H,N,l)$ in the same way as for the examples in Table~\ref{table:comparison_examples}. We observe that despite being successful, the computation time of the IRPI approach scales poorly with dimension of the output constraint set $\mathcal{Y}$. This is because the number of vertices $v_{\mathcal{Y}}=2^{n_y}$ leads to an exponential increase in the number of variables and constraints defining Problem~\eqref{eq:equivalent_SF_form_nonlinear} with $n_y$.
This issue can be tackled if we encode $\mathcal{Y} \subseteq \mathcal{S}(l,\mathcal{W}) \oplus \mathbb{B}(\epsilon)$ directly using the hyperplane representation of $\mathcal{Y}$, a subject deferred to future research.
\begin{table}
\centering
\resizebox{0.92\columnwidth}{!}{
\begin{tabular}{|c|c|c|c|}
\hline
$\#$  & $(n_x,n_w,n_y,s,m_X,m_W)$  & $\#$  & $(n_x,n_w,n_y,s,m_X,m_W)$ \\ \hline
$1$ & $(2,2,2,8,12,6)$ &
$2$ & $(2,2,3,15,12,6)$ \\
$3$ & $(2,3,2,12,18,42)$ &
$4$ & $(2,3,3,2,18,42)$ \\
$5$ & $(2,4,2,9,20,80)$ &
$6$ & $(2,4,3,17,8,80)$ \\
$7$ & $(3,2,2,13,34,6)$ &
$8$ & $(3,2,3,14,34,6)$ \\
$9$ & $(3,3,2,16,36,42)$ &
$10$ & $(3,3,3,16,35,42)$ \\
$11$ & $(3,4,2,7,34,80)$ &
$12$ & $(3,4,3,11,34,80)$ \\
$13$ & $(4,2,2,18,68,6)$ &
$14$ & $(4,2,3,4,8,6)$ \\
$15$ & $(4,3,2,18,74,42)$ &
$16$ & $(4,3,3,10,75,42)$ \\
$17$ & $(4,4,2,7,80,80)$ &
$18$ & $(4,4,3,13,53,80)$ \\
$19$ & $(5,2,2,7,67,6)$ &
$20$ & $(5,2,3,6,20,6)$ \\
$21$ & $(5,3,2,5,72,42)$ &
$22$ & $(5,3,3,6,102,42)$ \\
$23$ & $(5,4,2,5,118,80)$ &
$24$ & $(5,4,3,7,108,80)$  \\ \hline
\end{tabular}
}
\caption{Randomly generated systems}
\label{table:comparison_examples}
\end{table}

\begin{table}
\centering
\resizebox{0.92\columnwidth}{!}{
\begin{tabular}{|c|c|c|c|}
\hline
$\#$  & $(n_x,n_w,n_y,\rho(A),s,\norm{\epsilon_{\mathrm{IRPI}}}_1,\mathrm{t}_{\mathrm{IRPI}}\text{[s]},\mu_{\mathrm{IRPI}})$ \\ \hline
$1$ & $(10,5,2,0.6975,40,0.5166,45.4723,0.0073
)$  \\ 
$2$ & $(20,10,3,0.69,38,1.0529,258.9705,0.0062
)$  \\
$3$ & $(50,20,4,0.699,42,1.3626,1.1432\times 10^4,0.0064
)$  \\ 
$4$ & $(100,20,2,0.8,76,0.3093,197.892,0.0061
)$  \\\hline
\end{tabular}
}
\caption{Higher-dimensional systems}
\label{table:higher_dimension_systems}
\end{table}

\begin{figure}[t]
\centering
\hspace*{-0.4cm}
{
\resizebox{0.43\textwidth}{!}
{
\begin{tikzpicture}
\begin{scope}[xshift=0cm]
\node[draw=none,fill=none](tanks_fig) {\includegraphics[trim=5 105 0 5,clip,scale=3.5]{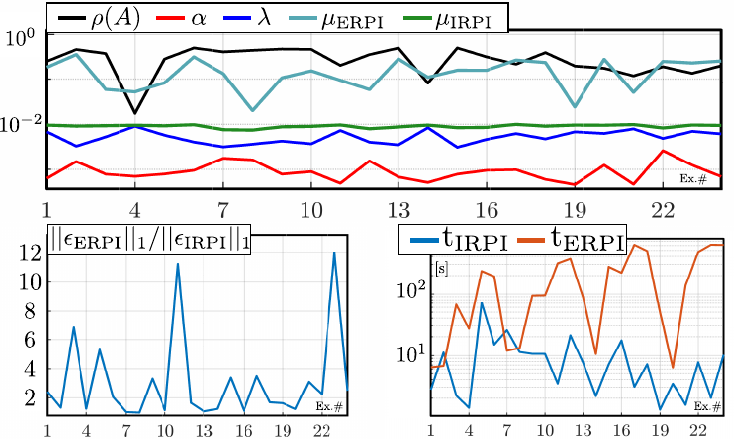}};
\end{scope}
\end{tikzpicture}}
}
\captionsetup{width=1\linewidth}
\caption{Comparison against the ERPI approach. The implicit RPI approach tackles Problem~\eqref{eq:orig_problem_to_solve} with reduced conservativeness and improved computational efficiency.}
\label{figure:comparison_ERPI}
\end{figure}

{\color{black} 
\subsection{Comparison with the approach of~\cite{DeSantis1994}}
\label{sec:deSantis}
\begin{figure}[t]
\centering
\hspace*{-0.4cm}
{
\resizebox{0.52\textwidth}{!}
{
\begin{tikzpicture}
\begin{scope}[xshift=0cm]
\node[draw=none,fill=none](tanks_fig) {\includegraphics[trim=20 0 0 0,clip,scale=3.5]{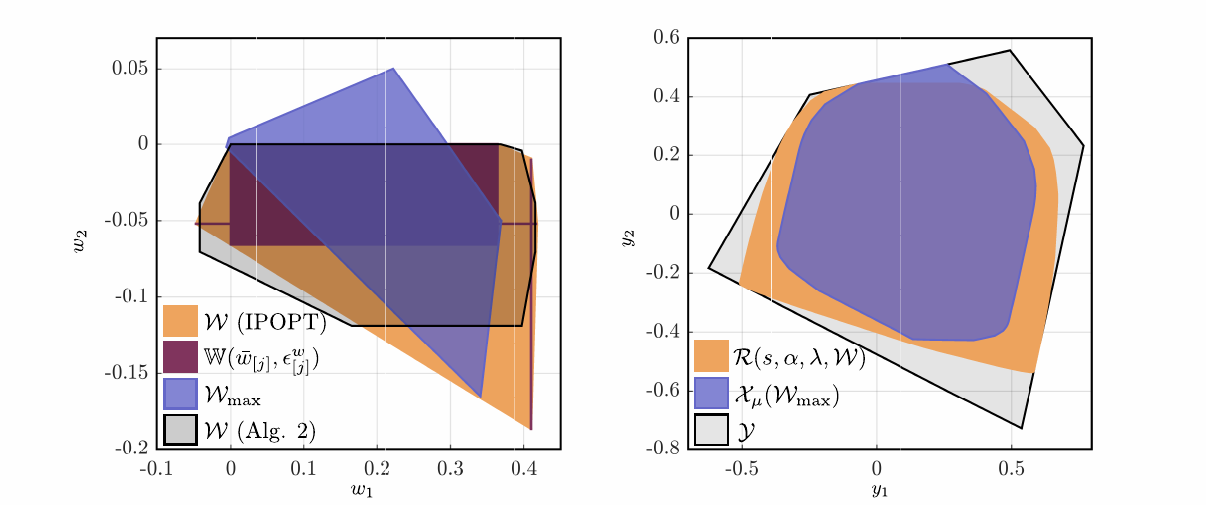}};
\end{scope}
\end{tikzpicture}}
}
\captionsetup{width=1\linewidth}
\caption{Comparison to \cite{DeSantis1994}: $\mathcal{W}$ yields $\norm{\epsilon}_1=0.545$ (vs. $0.8254$ as obtained with $\mathcal{W}_{\mathrm{max}}$ obtained using \cite{DeSantis1994}).}
\label{figure:compare_deSantis}
\end{figure}
We compare our approach with that of~\cite{DeSantis1994} for computing the maximal disturbance set. For simplicity, we consider LTI systems with $B,C=\I$ and $D=\0$ such that $\mathcal{Y}$ is the state constraint set. Their approach assumes to know a PI set $\mathcal{X}_{\mathrm{PI}}$ inside $\mathcal{Y}$ for the system $x(t+1)=Ax(t)$, and computes the maximal disturbance set as $\mathcal{W}_{\mathrm{max}}:=\mathcal{X}_{\mathrm{PI}} \ominus A\mathcal{X}_{\mathrm{PI}}$. Hence, they render $\mathcal{X}_{\mathrm{PI}}$ robust against the disturbance set $\mathcal{W}_{\mathrm{max}}$, and the inclusion $C\mathcal{X}_{\mathrm{m}}(\mathcal{W}_{\mathrm{max}}) \subseteq \mathcal{Y}$ holds since the mRPI set $\mathcal{X}_{\mathrm{m}}(\mathcal{W}_{\mathrm{max}})$ is included in the RPI set. Thus, $\mathcal{W}_{\mathrm{max}}$ is a feasible solution to Problem~\eqref{eq:orig_problem_to_solve}. As in~\cite{Kalabic2012}, we compute $\mathcal{X}_{\mathrm{PI}}$ using the algorithm in~\cite{Kolmanovsky1998} by choosing it to be the largest $\lambda$-contractive set for the system $x(t+1)=Ax(t)$ inside $\mathcal{Y}$.

For comparison, we consider an LTI system with matrix
$
A=\begin{bmatrix} -0.4918  &  0.3255 \\
0.3500   & 0.7231 \end{bmatrix},
$
and $\mathcal{Y}$ the same as in Section~\ref{example:illustrative_example}.
We begin by determining the largest $0.99$-contractive set within $\mathcal{Y}$ and use it to compute $\mathcal{W}_{\mathrm{max}}$. Next, we solve Problem~\eqref{eq:equivalent_SF_form_nonlinear} to obtain the set $\mathcal{W}$ using the approach outlined in Section~\ref{sec:approx_soln_methods}. We simulate Algorithm~\ref{alg:SAL_solving_procedure} with $\mu=10^{-6}$ and $\gamma=2$ to identify the parameters $(s,\alpha,\lambda)$ required to parameterize the RPI set $\mathcal{R}(s,\alpha,\lambda,\mathcal{W})$. We use $N=3$ boxes to parameterize $\mathcal{W}$ along with $l=s=113$ and $H=[\I \ -\I]^{\top}$ to define the cost function. At termination of Algorithm~\ref{alg:alg2} and IPOPT, we obtain $\norm{\epsilon}_1$ values of $0.746$ and $0.545$, respectively. We further compute the $\mu$-RPI set $\mathcal{X}_{\mu}(\mathcal{W}_{\mathrm{max}})$ using Lemma~\ref{lemma:original_Rakovic} with $\mu=10^{-6}$ for comparison. The resulting $\mathrm{d}_{\mathcal{Y}}(C\mathcal{X}_{\mu}(\mathcal{W}_{\mathrm{max}}))=0.8254$ confirms that our approach can solve Problem~\eqref{eq:orig_problem_to_solve} with reduced conservativeness relative to $\mathcal{W}_{\mathrm{max}}$. This results from the synthesis of both an RPI set and a corresponding disturbance set with the aim of maximizing reachability in our case, while these phases are decoupled in the synthesis of $\mathcal{W}_{\mathrm{max}}$. These sets are plotted in Figure~\ref{figure:compare_deSantis}.
%
%
%
}

{\color{black}
\subsection{Reduced-order controller synthesis}
\label{sec:Reduced_order_control}
We will now present an application of the methods to synthesize constraint sets for reduced-order control schemes.
We consider a system with state $\mathrm{x}=[x^{\top}_{[1]} \ \  x^{\top}_{[2]}]^{\top}$, input $u \in \mathcal{U}$ and constrained output $y \in \mathcal{Y}$, given by
\begin{align*}
\mathrm{x}(t+1) &= \begin{bmatrix}  \mathrm{f}(x_{[1]}(t),x_{[2]}(t),u(t)) \\ A_{[21]}x_{[1]}(t)+A_{[22]}x_{[2]}(t) \end{bmatrix}, \\ 
y(t)&=C_{[1]}x_{[1]}(t)+C_{[2]}x_{[2]}(t),
\end{align*}
i.e., with coupled dynamics and constraints between the substates $x_{[1]}$ and $x_{[2]}$. For such a system, we aim to compute a constraint set $\mathbb{X}_{[1]}$ over the substates $x_{[1]}$ that satisfies
\begin{align}
\label{eq:reduced_order_condition}
x_{[1]}(t) \in \mathbb{X}_{[1]} \Rightarrow y(t) \in \mathcal{Y}, && \forall \ t \geq 0.
\end{align}
This problem is encountered, for example, in reduced-order control~\cite{Kalabic2012},\cite{Sopasakis2013}, decentralized control~\cite{Mulagaleti2021},\cite{Riverso2013}, etc., in which a subset of states $x_{[2]}$ are inaccessible to the controller that selects the control input $u \in \mathcal{U}$, but safe operation must be ensured, i.e., $y \in \mathcal{Y}$. Then, the output constraints must be translated to constraints on the accessible states $x_{[1]}$. 
Under the assumption that $\rho(A_{[22]})<1$, a key observation towards computing such a set $\mathbb{X}_{[1]}$ is that if $x_{[1]}(t) \in \mathbb{X}_{[1]}$ for all $t\geq 0$, then substate $x_{[2]}$ is constrained to the corresponding mRPI set for the system $x_{[2]}(t+1)=A_{[22]}x_{[2]}(t)+A_{[21]}x_{[1]}(t)$, i.e., 
\begin{align*}
x_{[1]}(t) \in \mathbb{X}_{[1]} \Rightarrow x_{[2]}(t) \in \mathcal{X}_{\mathrm{m}}(\mathbb{X}_{[1]}):=\bigoplus_{t=0}^{\infty} A^t_{[22]}A_{[21]}\mathbb{X}_{[1]}
\end{align*}
for all $t\geq 0$. This implies that the desired set $\mathbb{X}_{[1]}$ satisfying~\eqref{eq:reduced_order_condition} can be computed by solving the optimization problem
\vspace{-15pt}
\begin{subequations}
\label{eq:reduced_order_problem}
\begin{align}
\min_{\0 \in \mathbb{X}_{[1]}} & \quad \mathrm{d}_{\mathcal{Y}}(C_{[2]}\mathcal{X}_{\mathrm{m}}(\mathbb{X}_{[1]})\oplus C_{[1]}\mathbb{X}_{[1]}) 
\\
\text{s.t.} &  \quad C_{[2]}\mathcal{X}_{\mathrm{m}}(\mathbb{X}_{[1]})\oplus C_{[1]}\mathbb{X}_{[1]} \subseteq \mathcal{Y}.
\end{align}
\end{subequations}
The problem at hand is similar to Problem~\eqref{eq:orig_problem_to_solve} with the disturbance set $\mathcal{W}$ represented by $\mathbb{X}_{[1]}$. Hence, it can be tackled using the methods presented in this paper, following which feedback controllers can be designed for the subsystem $x_{[1]}(t+1)=\mathrm{f}(x_{[1]}(t),\tilde{w}(t),u(t))$
with constraints $x_{[1]} \in \mathbb{X}_{[1]}$, $u \in \mathcal{U}$, and disturbances $\tilde{w} \in \mathcal{X}_{\mathrm{m}}(\mathbb{X}_{[1]})$. These controllers will ensure that the system constraints $y \in \mathcal{Y}$ are always satisfied.
\begin{remark}
In order to successfully design such robust controllers, the set $\mathbb{X}_{[1]}$ must at least contain a robust control invariant (RCI) set for the controlled subsystem under the action of disturbances $\tilde{w} \in \mathcal{X}_{\mathrm{m}}(\mathbb{X}_{[1]})$, i.e., the condition
\begin{align}
\label{eq:RCI_condition}
\exists \ \ \begin{matrix} x_{[1]} \in \mathbb{X}_{[1]}, \ u \in \mathcal{U} \end{matrix} \ \  : \ \  \mathrm{f}(x_{[1]},\mathcal{X}_{\mathrm{m}}(\mathbb{X}_{[1]}),u) \subseteq \mathbb{X}_{[1]}
\end{align}
%
must be satisfied. Enforcing condition~\eqref{eq:RCI_condition} as a constraint in Problem~\eqref{eq:reduced_order_problem} is a subject of future research. It is worth remarking that in the decentralized MPC scheme proposed in~\cite{Mulagaleti2020}, such constraints were successfully incorporated in the procedure to compute $\mathbb{X}_{[1]}$. This was facilitated by the use of explicit RPI sets to approximate mRPI set $\mathcal{X}_{\mathrm{m}}(\mathbb{X}_{[1]})$.  $\hfill\square$
\end{remark}

As a numerical example, we consider an LTI system, i.e., with $\mathrm{f}(x_{[1]},x_{[2]},u)=A_{[11]}x_{[1]}+A_{[12]}x_{[2]}+B_{[1]}u$, where
\begin{align*}
&\scalemath{0.85}{
\begin{bmatrix}
\overset{A_{[11]}}{\overbrace{
\begin{matrix*}[r]
1.0000  &  1.0000 \\ 
0  &  1.0000 \end{matrix*}}}
& \vline &
\overset{A_{[12]}}{\overbrace{
\begin{matrix*}[r] -0.0524 &  -0.3299 &   0.3061  &  0.2773 \\ -0.0048  & -0.1020  &  0.1244  & -0.1044\end{matrix*} }} 
\vspace{2pt} \\ \hline \vspace{-4pt} \\
\underset{A_{[21]}}{\underbrace{
\begin{matrix*}[r]
0   & 0.0204 \\
0   & 0.0344 \\
0   &-0.0339 \\
0   & 0.0134	
\end{matrix*}}} & \vline &
\underset{A_{[22]}}{\underbrace{
\begin{matrix*}[r] 
-0.0790  &  0.2854 &  -0.0377  &  0.6949 \\
0.2854  & -0.2284   & 0.2752   & 0.3536 \\
-0.0377  &  0.2752 &   0.6021  & -0.2824 \\
0.6949   & 0.3536 &  -0.2824  & -0.0129
\end{matrix*} 
}}	
\end{bmatrix},}
\vspace{5pt}\\
&\scalemath{0.85}{
\begin{bmatrix}
\underset{C_{[1]}}{\underbrace{
\begin{matrix*}[r]
0.9407 &  -0.3282 \\
-0.6624&   -0.7257
\end{matrix*}
}}
& \vline &
\underset{C_{[2]}}{\underbrace{
\begin{matrix*}[r] 
0.8716  &  0.3587   & 0.2407 &   0.5116 \\
-0.1863 &   0.1624 &   0.7122 &   1.7494
\end{matrix*}
}}
\end{bmatrix},}
\end{align*}
$B_{[1]}=[0.5 \quad 1]^{\top}$. For this system, we synthesize a set $\mathbb{X}_{[1]}$ by solving Problem~\eqref{eq:reduced_order_problem}. To this end, we use the procedure presented in this paper with parameters $s=150$, $\lambda=2.931 \times 10^{-5}$, $\alpha=5.195 \times 10^{-4}$ that are computed by Algorithm~\ref{alg:SAL_solving_procedure} 
for $\mu=10^{-3}$ and $\gamma=0.1$ in $0.44$s. We parametrize the set $\mathbb{B}(\epsilon)$ used to define $\mathrm{d}_{\mathcal{Y}}(\cdot)$ in Equation~\eqref{eq:objective_exact} with $n_B=8$ vectors $\{H_i^{\top}, i \in \mathbb{I}_1^{n_B}\}$ sampled uniformly from the surface of $\mathcal{B}_{\infty}^2$. Finally, we parametrize the disturbance set in \eqref{eq:convex_hull_of_boxes} with $N=4$ boxes in $\R^2$, and select $l=s=150$ to formulate constraint \eqref{eq:final_formulation_to_solve:con4}. In order to solve the resulting Problem~\eqref{eq:equivalent_SF_form_nonlinear}, we first implement Algorithm~\ref{alg:alg2}. The algorithm was initialized with each $\beta^1_{[itj]},\beta^2_{[ij]}=1/N$, and it terminates in $\iota=16$ iterations with $\norm{\epsilon}_1=0.9726$, consuming an average of $0.7655$s per iteration. We then use the output of the algorithm to initialize IPOPT to solve Problem~\eqref{eq:equivalent_SF_form_nonlinear}. At termination, we obtain $\norm{\epsilon}_1=0.8828$. In Figure~\ref{figure:RMPC_plot1}, we show the sets obtained at the termination of Algorithm~\ref{alg:alg2} and IPOPT. 
\begin{figure}[t]
\centering
\hspace*{-0.85cm}
{
\resizebox{0.58\textwidth}{!}
{
\begin{tikzpicture}
\begin{scope}[xshift=0cm]
\node[draw=none,fill=none](tanks_fig) {\includegraphics[trim=20 0 0 0,clip,scale=1]{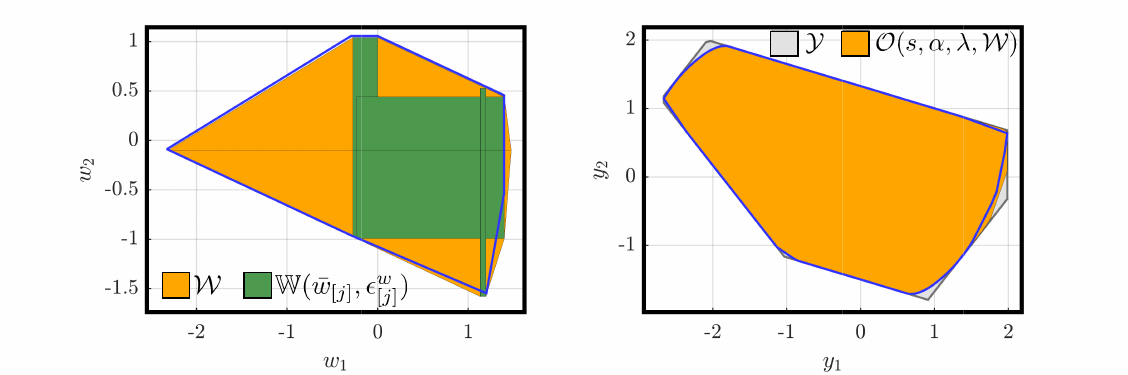}};
\end{scope}
\end{tikzpicture}}
}
\captionsetup{width=1\linewidth}
\caption{Results obtained by approximating Problem~\eqref{eq:reduced_order_problem} as Problem~\eqref{eq:equivalent_SF_form_nonlinear} to compute $\mathbb{X}_{[1]}=\mathcal{W}$. The blue lines indicate the boundaries of the corresponding sets obtained at termination of Algorithm~\ref{alg:alg2}. The set $\mathcal{Y}$ is a randomly generated polytope.}
\label{figure:RMPC_plot1}
\end{figure}

Using the computed set $\mathbb{X}_{[1]}$, we synthesize a tube-based robust model predictive controller~\cite{Mayne2005} for the subsystem $x_{[1]}(t+1)=A_{[11]}x_{[1]}(t)+B_{[1]}u(t)+A_{[12]}\tilde{w}(t)$ with state $x_{[1]} \in \mathbb{X}_{[1]}$, input $u \in \{u : \norm{u}_{\infty} \leq 0.7\}$, and disturbances $\tilde{w} \in \mathcal{R}(s,\alpha,\lambda,\mathbb{X}_{[1]}) \supseteq \mathcal{X}_{\mathrm{m}}(\mathbb{X}_{[1]})$. The goal of the controller is to track references over the output $[1 \ 0]x_{[1]}$. For details regarding this implementation, we refer the reader to~\cite[Chapter 5]{MulagaletiThesis}. In Figure~\ref{figure:RMPC_plot2}, we show the closed-loop performance of this controller. As expected, we guarantee constraint satisfaction for the overall system by enforcing $x_{[1]} \in \mathbb{X}_{[1]}$. 	
\begin{figure}[t]
\centering
\hspace*{-0.75cm}
{
\resizebox{0.55\textwidth}{!}
{
\begin{tikzpicture}
\begin{scope}[xshift=0cm]
\node[draw=none,fill=none](tanks_fig) {\includegraphics[trim=20 0 0 0,clip,scale=1]{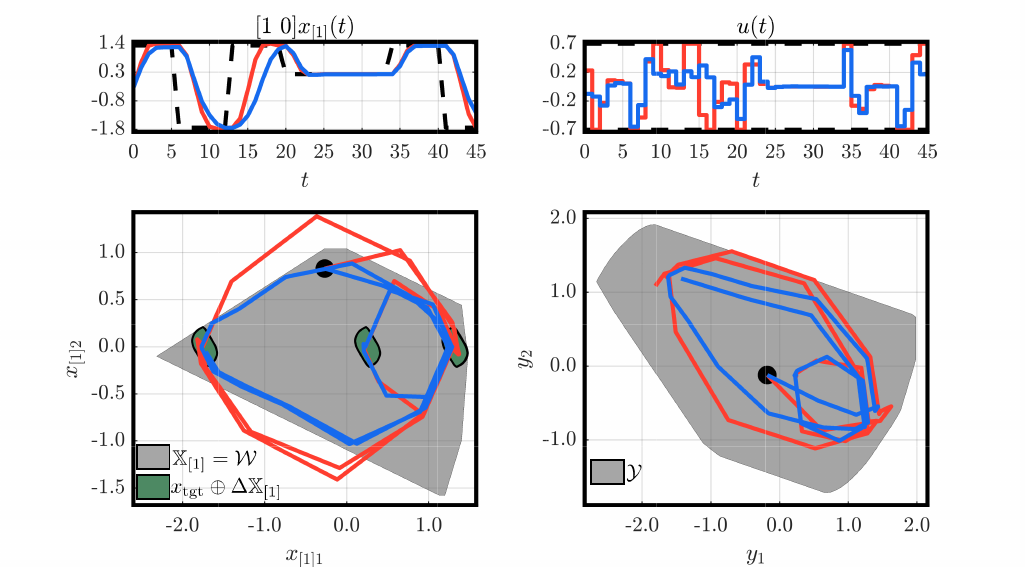}};
\end{scope}
\end{tikzpicture}}
}
\captionsetup{width=1\linewidth}
\caption{The blue lines in the bottom plots depict closed-loop performance of the reduced-order robust MPC controller. The set $\Delta \mathbb{X}_{[1]}$ is the RPI set for the controlled system in tube MPC~\cite{Mayne2005}. By enforcing $x_{[1]} \in \mathbb{X}_{[1]}$, constraints $y \in \mathcal{Y}$ are satisfied without measurements of the state $x_{[2]}$. The red lines indicate the closed-loop performance of a full-order MPC controller, which has measurements of $x_{[2]}$. While the reduced-order scheme is more conservative, it ensures safe operation without requiring $x_{[2]}$. The black dots indicate $x_{[1]}(0)$
and $y(0)$ in bottom-left and bottom-right plots respectively.}
\label{figure:RMPC_plot2}
\end{figure}	
}
\section{Conclusions}
We presented a method to compute a disturbance set for a discrete-time LTI system such that the output-reachable set is contained inside an assigned output constraint set, while maximizing its coverage. To this end, we formulated an optimization problem using implicit RPI sets that provide a priori approximation error guarantees, and used a novel disturbance set parameterization that permitted the encoding of feasible disturbance sets as a polyhedron. We then present a solution method for the resulting optimization problem. 
We demonstrated that the proposed approach outperforms state-of-the-art, both with respect to conservativeness and computational efficiency.
Future research will focus on: ($a$) Extending the technique to accomodate linear difference inclusions \cite{Kouramas2005}, ($b$) Encoding inclusion~\eqref{eq:final_formulation_to_solve:con4}, i.e., $\mathcal{Y} \subseteq \mathcal{S}(l,\mathcal{W}) \oplus \mathbb{B}(\epsilon)$ directly using a hyperplane representation of $\mathcal{Y}$.

\bibliography{references}
\end{document}